\documentclass
[a4paper]
{article}
\usepackage[a4paper, margin=2.5cm]{geometry}

\fontsize{11pt}{13.6pt}\selectfont

\usepackage[T1]{fontenc}
\usepackage{lmodern}

\usepackage{amsthm}
\usepackage{xspace}
\usepackage{tikz}
\usepackage{subcaption}
\usetikzlibrary{shapes.geometric}
\usepackage{ulem}
\usepackage{amsmath}
\usepackage{hyperref}
\usepackage{cleveref}

\newtheorem{definition}{Definition}
\newtheorem{algorithm}{Algorithm}
\newtheorem{lemma}{Lemma}
\newtheorem{theorem}{Theorem}
\newtheorem{corollary}{Corollary}
\newtheorem{example}{Example}

\newtheorem{examplegccont}{Example Colouring: Continuation}
\newtheorem{examplemaxcont}{Example Max: Continuation}

\definecolor{darkgreen}{rgb}{0.0,0.5,0.0}

\renewcommand{\dag}{DAG\xspace}

\newcommand{\dc}{DC\xspace}

\newcommand{\gc}{GC\xspace}

\newcommand{\imped}{impedensable\xspace}
\newcommand{\Imped}{Impedensable\xspace}
\newcommand{\ldag}{\ensuremath{\prec}-\dag}

\renewcommand{\sp}{SP\xspace}
\newcommand{\po}{PO\xspace}

\title{
Properties of
Concurrent Algorithms that Allow Asynchrony} 

\author{Arya Tanmay Gupta$~^{1,2}$\and Sandeep S Kulkarni$~^2$}

\date{
$~^1$RLVS Aonla, India\\$~^2$Computer Science and Engineering, Michigan State University, USA\\ \texttt{\{atgupta,sandeep\}@msu.edu}
}

\begin{document}

\maketitle

\begin{abstract}
In this paper, we investigate the conditions under which a concurrent system is able to guarantee convergence deterministically even without synchronization. In general, an algorithm is deemed asynchrony-tolerant iff the global states form a DAG even if the nodes can execute based on outdated/stale values from each other. In this paper, we show that the induction of a partial order (PO) among the local states visited by individual nodes is necessary and sufficient to allow asynchrony.

There are certain problems for which the predicate that defines the problem, itself, induces a partial order among the local states (called \textit{PO-inducing problems}). On the other hand, \textit{non-PO-inducing problems} is another class of problems in which a predicate defining the problem cannot induce a partial order among local states. We introduce these classes of problems. We study example PO-inducing and non-PO-inducing problems, and provide example PO-inducing (in other words, asynchrony-tolerant) algorithms (algorithms under which any node traverses a partial order of local states) for all these problems.
\end{abstract}

\section{Introduction} 

Consider the graph colouring (\gc) problem. An algorithm for this problem can be developed as follows: if some node $i$ reads that it has a conflicting colour with any of its neighbours, then it changes its colour to the minimum possible available colour.
If this algorithm is run under a central scheduler (where only one node executes at a time), then it will converge to a state where all nodes have non-conflicting colours. 

In a uniprocessor execution (or a multiprocessor execution under a central scheduler), 
this algorithm runs correctly because the global states form a directed acyclic graph (\dag), and all the sink nodes of this \dag are optimal states.
This property is necessary and sufficient for the correctness of any 
algorithm that must reach an optimal state and then terminate/stutter.
However, if the above algorithm is run on a multiprocessor system in asynchrony, it may never converge.
This can be explained by the following example. Consider a graph with two nodes, $i$ and $j$, connected to each other by an edge. Let that both $i$ and $j$ are initialized with colour value 1. Suppose that $i$ reads the local state of $j$ and decides to change its colour to 2. However, if $j$ reads the state of $i$ asynchronously before $i$ finalizes the change in its colour, $j$ will also decide to change its colour to 2. Similarly, in the next step, both nodes may decide to change their colour to 1. This can continue forever and the system may never converge.
This is because of the presence of cycles in the global state transition graph, which, in turn, are induced due to the race conditions arising among the neighbouring nodes. To prevent this, synchronization techniques such as fencing (cf. \cite{Trachsel2006}) are deployed.

From above, we observe that the assumption of synchronization makes the analysis of system and design of algorithms easier. Additionally, if such assumptions are removed, then we may observe cyclic behaviour in state transitions which prevents the system from convergence. Thus, designing an algorithm that guarantees convergence without synchronization is not trivial.
However, enforcing synchronization is costly in terms of time and computational resources, so the community desires to develop algorithms with minimum possible synchronization assumptions. 

If an algorithm does not require synchronization, then it can highly benefit from the available concurrency, without having to block the execution of a node for the sake of synchronization.
If we understand the underlying behaviour of such algorithms, then we can easily analyze if an algorithm possesses this quality, and if not, we can make minimal changes to it and transform it into an algorithm that does not require synchronization.
With this motivation, in this paper, we study the properties of an algorithm that allow it to converge even when the nodes execute in asynchrony.

This paper is built on the work by Garg et al. on lattice-linear problems and by Gupta and Kulkarni on lattice-linear problems and non-lattice-linear problems. We provide citations to their work in the section on related work (\Cref{section:literature}).

\subsection{Contributions of the paper}

We introduce two major classes of problems: partial order-inducing (\po-inducing) problems (with include dominant clique problem, the shortest path problem) and non-\po-inducing problems (which includes maximal matching (MM) problem).

Algorithms are presented for these three problems.
We show how PO-inducing problems allow asynchrony naturally. It means that in this class of problems, the predicate that defines a subject problem can be used by the nodes to determine if they should change their local state. Conversely, in non-PO-inducing systems, the problem definition does not identify which nodes should move. We show that algorithms can be designed even for non-PO-inducing problems by restricting state transitions.

The upper bound to the convergence time of a \po-inducing algorithm is studied.
It is shown that inducing a partial order in the local state space is necessary and sufficient to allow asynchrony.

\subsection{Organization of the paper}

We discuss the preliminaries in \Cref{section:preliminaries}. In \Cref{section:dip}, we study the characteristics of \po-inducing problems, and
in \Cref{section:dia}, we study the characteristics of \po-inducing algorithms, with examples.
In \Cref{section:dag-properties}, we study the properties of \po-inducing systems; while the previous sections provide simple, but sufficient, examples of asynchrony-tolerant systems, this section is crucial from the perspective of the theory that we establish in this paper.
We discuss related work in \Cref{section:literature}. We conclude in \Cref{section:conclusion}. 


\section{Preliminaries}\label{section:preliminaries}

The algorithms that we present in this paper are graph algorithms where the input is a graph $G$, $V(G)$ is the set of its nodes, $E(G)$ is the set of its edges, $n=|V(G)|$, and $m=|E(G)|$. For a node $i\in V(G)$, $Adj_i$ is the set of nodes adjacent to $i$, and $Adj^x_i$ is the set of nodes within distance-$x$ from $i$, excluding $i$. $N_i=Adj_i\cup\{i\}$ and $N_i^x=Adj_i^x\cup\{i\}$.
Degree of a node $i$ is the number of nodes connected to $i$, i.e., $deg(i)=|Adj_i|$. The length of shortest path from $i$ to node $j$ is denoted by $dis(i,j)$. 

Each node $i$ is associated with a set of variables. The algorithms are written in terms of rules, where each \textit{rule} is of the form $g\longrightarrow a_c$, where the \textit{guard} $g$ is a proposition consisting of the variables of $i$ along with the variables of other nodes.
We design $g$ over variables in $N_i^x$, where $x\geq 1$ is problem specific.
If any of the guards holds true for $i$, we say that $i$ is \textit{enabled}.
If $i$ is enabled, $i$ makes a \textit{move} by updating its variables by executing \textit{action} $a_c$ corresponding to the guard $g$ that holds true.
A \textit{round} is a sequence of events in which each node evaluates its guards at least once, and makes a move if it is enabled.

We denote $S$ to be the set of all global states.  A global state $s\in S$ is represented as a vector where $s[i]$ denotes the local state of node $i$. The local state of node $i$, $s[i]$, is
itself 
a vector of the variables of $i$.
The \textit{state transition system} $\mathcal{S}$ on the state space $S$ is a discrete structure that defines all the possible transitions that can take place among the states in $S$. Under a given algorithm $A$, $\mathcal{S}$ is a directed graph such that $V(\mathcal{S})=S$, and $E(\mathcal{S})=\{\langle s, s'\rangle | \langle s, s'\rangle$ is a state transition under $A\}$.

An algorithm $A$ is \textit{self-stabilizing} with respect to the subset $S_o$ of $S$ iff it satisfies two conditions: (1) \textit{convergence}: starting from an  arbitrary state, any sequence of computations of $A$ reaches a state in $S_o$, and (2) \textit{closure}: any computation of $A$ starting from $S_o$ always stays in $S_o$. 
We assume $S_o$ to be the set of \textit{optimal} global states: the system is deemed converged once it reaches a state in $S_o$. All other states are \textit{suboptimal}. The algorithm $A$ is \textit{silent} if no node becomes enabled once a state in $S_o$ is reached.

\subsection{Computation and Communication Model}\label{subsection:comcom}

In this paper, 
we are interested in the systems in which nodes execute without coordination. Nodes can read/obtain the local states of each other, which they may require to take an action. However, they execute asynchronously, i.e., they are allowed to execute based on outdated/stale values from each other.

The sequence of global states that the system traverses through can be of the form $s_0,s_1,s_2,...$ . If some node $i$ in global state $s_{\ell}$ executes, and the information it reads about some node $j$ is possibly stale, then the information that $i$ has about the local state of $j$ can be any local state from $\{s_0[j],s_1[j],...,s_{\ell}[j]\}$.
Asynchronous systems assume that stale values can be read in this manner and despite this, they guarantee convergence.
We explain execution without synchronization in more detail in Appendix \ref{appendix:computation-model}.

We are interested in both communication models: shared memory model (where nodes have direct access to the memory and variables of each other) and message passing model (where nodes can only obtain values from each other via sending messages to each other). We discuss in depth on the description of communication models in Appendix \ref{appendix:asynchrony-types}; this is to explain the extent of the weakness of the assumption of asychrony. We assume no synchronization among the nodes. However, there are two variations of asynchrony that we are interested in. One is arbtirary asynchrony (AA) that is defined in the above paragraph.
Another variation is asynchrony with monotonous read (AMR model). 
Similar to AA, in AMR, the nodes execute asynchronously.
However, AMR model adds an additional, but weak, restriction: 
a node can read the local states of other nodes only in the order of their transition.
E.g., if the state of node $j$ changes from 0 to 1 to 2 and node $i$ reads the state of $j$ to be 1 then its subsequent read will either return 1 or 2, it cannot return 0. We discuss examples to help explain asynchrony in Appendix \ref{appendix:colouring-and-max}.

\section{Natural partial order induction: \po-inducing \textit{Problems}}
\label{section:dip}

In this section, we discuss properties of problems where a partial order (PO) among the local states visited by individual nodes is 
induced \textit{naturally}.
It means that in any suboptimal state, the predicate that defines the problem can be used by the nodes to determine if that must change their state.

With respective definitions, we elaborate some examples in Appendix \ref{appendix:examples-following-definitions}. They will be helpful for a reader to understand the respective definitions that we present in the following.

\subsection{Embedding a \texorpdfstring{\ldag}~among global states}\label{subsection:<-dag}

We first define a partial order $\prec_l$ among the local states of a node.
This partial order 
defines
all possible transitions that a node is allowed to take.
The partial order $\prec_l$ is used to restrict how node $i$ can execute: $i$ can transit from state $s[i]$ to $s'[i]$ only if $s[i] \prec_l s'[i]$.

Using 
$\prec_l$,
we define $\prec_g$ that orders the global states.
The predicate $s \prec_g s^\prime$ is true iff the predicate $(\forall i:$ $s[i]=s'[i]\lor s[i]\prec_l s'[i]) \land (\exists i:s[i]\prec_ls'[i])$ is true; $s=s'$ iff $\forall i:s[i] = s'[i]$. 
For brevity, we use $\prec$ to denote $\prec_l$ and $\prec_g$: $\prec$ corresponds to $\prec_l$ while comparing local states, and $\prec$ corresponds to $\prec_g$ while comparing global states. 
We also use the symbol `$\succ$' which is such that $s\succ s'$ iff $s' \prec s$.
Similarly, we use symbols `$\preceq$' and `$\succeq$'; e.g., $s\preceq s'$ iff  $s=s' \lor s \prec s'$.
We call the \dag, formed from $\prec_g$, a \textit{\ldag}.

\begin{definition}\label{definition:<-dag}
    \textbf{\boldmath \ldag}. 
    Given a partial order $\prec_l$ that orders the local states visited by $i$ (for each $i$), the \ldag corresponding to $\prec_l$ is defined as follows:
    $s \prec s'$ iff $(\forall i: s[i] \preceq s'[i]) \wedge (\exists i: s[i] \prec s'[i])$.
\end{definition}

A \ldag constraints the transition among global states: state $s$ can transition to state $s'$ iff $s\prec s'$.
Due to a \ldag embedded in the state space, the algorithm can execute asynchronously.
Under proper constraints on the structure of \ldag, convergence can be ensured. We elaborate on this in \Cref{subsection:properties-dip}. 

\noindent\textbf{\textit{Remark}}: If local states form a total order under $\prec_l$, which is a special case of a partial order, then the global states form a $\prec$-lattice under $\prec_g$, which is a special case of a \ldag.

\subsection{General Properties of \po-Inducing Problems}\label{subsection:properties-dip}

In the above subsection, we described the structural properties of a system in which the local state transition graph for each node forms a partial order.
Next, we discuss what the importance of a partial order is, in defining the behaviour of a multiprocessing system.
We are not interested in systems where the nodes play against each other to maximize their profit,
but 
in systems where nodes execute such that the system as a whole reaches an optimal state. Only, they do not have to coordinate with each other.

Let us consider the imposition of asynchronous executions from the perspective of a single node $i$. Node $i$ can ensure correctness and convergence to an optimal local state while moving asynchronously only if it does not have to cycle back to any of its previously obtained local states. Consequently, a partial order would be induced among the local states, and a \ldag is induced among the global states.

Further, $i$ can ensure that it does not have to cycle back only by ensuring that it makes a move only when its current local state will hinder the system from reaching \textit{any} optimal global state.
Then, it will not matter when $i$ executes, it only matters that $i$ must move, implying that $i$ can execute asynchronously.

The above two paragraphs lay down the foundation of our general theory. If algorithms are designed in a way that all nodes can make this evaluation, then all nodes can execute asynchronously.

\subsubsection{\textbf{\Imped Node and \Imped Global State}}

Let $P$ be an arbitrary problem, and let that a system can be designed to solve $P$ such that a predicate $\mathcal{P}$ can be used to identify if the nodes have to make a move. We call such nodes \textit{\imped nodes} (\textit{indispensable} to move for progress, an \textit{impediment} to progress if it does not change its state).

This means that for any given global state $s$, if no node has to change its state, then $s$ is an optimal state. We say that $\mathcal{P}(s)$ holds true iff $s$ is an optimal state. Thus, an \imped node is formally defined as follows.


\begin{definition}\label{definition:imped-node}\textbf{Impedensable node.} $\textsc{\Imped}(i,s,\mathcal{P})\equiv \lnot \mathcal{P}(s)\land(\forall s'\succeq s:s'[i]=s[i]\Rightarrow\lnot \mathcal{P}(s'))$.
\end{definition}

\noindent\textbf{\textit{Remark}}: Since local states form a partial order, when $i$ moves, then $i$ may have multiple choices of actions. Similarly, when $i$ rejects a local state $s[i]$ and moves up in the partial order, then it may inadvertently reject other local states as well, because there is no longer a path to reach them through that partial order.

All suboptimal global states are \textit{\imped global states}. A global state $s$ is suboptimal, that is, $s$ is \imped, iff it contains at least one \imped node under predicate $\mathcal{P}$. Formally,

\begin{definition}\textbf{\textit{Impedensable global state.}}\label{definition:imped-state}
    $\textsc{\Imped}(s,\mathcal{P})\equiv \exists i:\textsc{\Imped}(i,s,\mathcal{P})$.
\end{definition}

A predicate $\mathcal{P}$ that is used by the nodes and global states of a system to determine if they are \imped is a \textit{\po-inducing predicate}. Formally,


\begin{definition}\textbf{\po-Inducing Predicate.}\label{definition:po-inducing-predicate}
    $\mathcal{P}$ is a \po-inducing predicate with respect to a \ldag, $\mathcal{S}$, induced among the global states iff $\forall s\in S: \lnot\mathcal{P}(s) \Rightarrow \exists i:\textsc{\Imped}(i,s,\mathcal{P})$.
\end{definition}

\noindent\textbf{\textit{Remark}}: If in each move the nodes have only one choice of action, then they also reject only one state in each move. In such a system, local states form a total order and effectively, the global states form a $\prec$-lattice. Such systems are called lattice-linear systems (cf. \Cref{section:literature} for citations).

\subsubsection{\textbf{Self-Stabilizing Predicates}}

All states that a global state $s$ can transition to are its successors. Formally,

\begin{definition}\textbf{\textit{Successors of a global state}}.
    $\textsc{Succ}(s)\equiv \{s':s'\succeq s\}$.
\end{definition}

\noindent Some global states may not have any successors. These are the topmost states in a \ldag. Formally,

\begin{definition}\label{definition:terminal-successors}\textbf{\textit{Terminal Successors}}.
    $\textsc{Terminal-Succ}(s,s') \equiv s'\in \textsc{Succ}(s) \land \textsc{Succ}(s')=\phi$.
\end{definition}



$\mathcal{P}$ is 
a
self-stabilizing 
predicate 
only if starting from any arbitrary state, the system converges to an optimal state. This, in turn, 
is possible only if all terminal successors in $\mathcal{S}$ are optimal states.
$\mathcal{P}$ can be true in other states as well, but in that case, the acting algorithm will not be silent. 

\begin{definition}\label{definition:ss-dip}\textbf{Self-stabilizing \po-inducing predicate}.
    Continuing from \Cref{definition:po-inducing-predicate}, $\mathcal{P}$ is a self-stabilizing \po-inducing predicate iff all terminal successors in the \ldag, $\mathcal{S}$, induced by $\mathcal{P}$ are optimal states, i.e.,
    $\forall s,s'\in S: \textsc{Terminal-Succ}(s,s')\Rightarrow \mathcal{P}(s')=true$.
\end{definition}

\subsubsection{\textbf{PO-inducing problems}}


There are certain problems for which a predicate $\mathcal{P}$, that defines a subject problem, can also be used to identify \imped nodes and \imped global states. We call such problems \textit{\po-inducing problems}. Formally, for a problem $P$ and a predicate $\mathcal{P}$ that defines $P$, we have the following.

\begin{definition}\label{definition:dip}
\textbf{\po-inducing problem}.
A problem $P$ is \po-inducing
iff there exists a predicate $\mathcal{P}$ and a \ldag, $\mathcal{S}$, is induced among the global states such that
(1) $P$ is solved iff the system reaches a state where $\mathcal{P}$ is true,
(2) $\mathcal{P}$ is \po-inducing with respect to $\mathcal{S}$, i.e., $\forall s: \neg \mathcal{P}(s) \Rightarrow \exists i:\textsc{\Imped}(i,s,\mathcal{P})$, and
(3) $\forall s:(\forall i: \textsc{\Imped}(i,s,\mathcal{P})\equiv \lnot \mathcal{P}(s)\land(\forall s':(s'\succeq s)\Rightarrow(s'[i]\in \textsc{Equicorrupt}(s[i],i)\Rightarrow\lnot \mathcal{P}(s'))))$.
\end{definition}

The main intent of this paper is to establish some fundamental properties of \po-inducing systems.
In the following, we study some \po-inducing problems, and explore how the local states forming a partial order
allows the nodes to execute 
without synchronization.

\subsection{Dominant Clique (\dc) problem}\label{subsection:dc}

In this subsection, we present an algorithm for the dominant clique problem. This is an \textit{embarrassingly parallel algorithm}, i.e., 
in this algorithm, 
there is no transfer of data among the nodes. However, we use this trivial algorithm to build the intuition behind the induction of a partial order among the local states visited by individual nodes, and effectively a \ldag among the global states.

\begin{definition}\textbf{Dominant clique}.
    In the dominant clique problem, the input is an arbitrary graph $G$ and every node stores the variable $i[cliq]$. The task is to (re-)evaluate $i[cliq]$ such that (1) all the nodes in $i[cliq]$ form a clique and $i\in i[cliq]$,
    and (2) there exists no clique $C$ in $G$ such that $i[cliq]$ is a proper subset of $C$.
\end{definition}

\noindent Thus, the dominant clique problem can be defined by the following predicate.
\begin{center}
    $\mathcal{P}_{dc}\equiv \forall i: (i\in i[cliq])\land(\forall j,k\in i[cliq]: (j\neq k)\Rightarrow(k\in Adj_j))\land$\\ $(\not\exists j\in Adj_i: j\not\in i[cliq]\land (\forall k\in i[cliq]:k\in Adj_j))$.
\end{center}
The local state of a node $i$ is defined by $\langle i[cliq]\rangle$. An \imped node $i$ in a state $s$ is a node for which (1) all the nodes in $i[cliq]$ do not form a clique, or (2) there exists some node $k$ in $Adj_i$ such that $i[cliq]\cup\{k\}$ is a valid clique, but $k$ is not in $i[cliq]$. Formally,
\begin{center}

$\textsc{\Imped-DC}(i)\equiv$ $(i\not\in i[cliq])\lor (\exists j,k\in i[cliq]:j\neq k\land k\not\in Adj_j))$ $\lor$\\ $(\exists j\in Adj_i:j\not\in i[cliq]\land (\forall k\in i[cliq]:k\in Adj_j))$.\\
\end{center}
The algorithm that we present is a self-stabilizing algorithm.
Thus, $i[cliq]$ may initially contain the nodes that are not neighbours of $i$. The algorithm is defined as follows. If all the nodes in $i[cliq]$ do not form a clique, then $i[cliq]$ is reset to be $\{i\}$. If there exists some node $j$ in $Adj_i$ such that $i[cliq]\cup\{j\}$ is a clique, but $j$ is not in $i[cliq]$, then $j$ is added to $i[cliq]$. 

\begin{algorithm}\label{algorithm:dc}
    Rules for node $i$.
\end{algorithm}
\begin{center}
    $\begin{array}{|l|}
        \hline
        \textsc{\Imped-DC}(i)\longrightarrow\\
        \begin{cases}
            i[cliq]=\{i\}
            \quad\quad\quad\quad\quad\quad
            \text{if $(i\not\in i[cliq]\lor(\exists j,k\in$}
            \text{$~i[cliq]:j\neq k\land k\not\in Adj_j))$}\\
            i[cliq]= i[cliq]~\cup\{j\}
            \quad\quad\text{otherwise (where $j$ is such that } 
            \text{$\forall k\in i[cliq]: k\in Adj_j$)}
        \end{cases}~\\
        \hline
    \end{array}$
\end{center}

\begin{lemma}\label{lemma:dc}
    The dominant clique problem is a \po-inducing problem.
\end{lemma}
\begin{proof} 
    For a node $i$, $i[cliq]$ contains the nodes that $i$ is connected with, and the nodes in $i[cliq]$ should form a clique. A global state does not manifest a dominant clique if at least one node $i$ in $s$ does not store a set of nodes forming a maximal clique with itself, i.e. (1) $i[cliq]$ is not a maximal clique, that is, there exists a $j$ in $Adj_i\setminus i[cliq]$ such that $i[cliq]\cup\{j\}$ forms a valid clique, or (2) the nodes in $i[cliq]$ do not form a clique.
    
    Next, we need to show that if some node $i$ in state $s$ is \imped, then for any global state $s'$ such that $s'\succeq s$, if $s'[i]=s[i]$, then $s'$ will not manifest a dominant clique. This is straightforward from the definition of \textsc{\Imped}-DC$(i)$ itself, that if a node $i$ is \imped, then $i$ does not store a set of nodes forming a maximal clique with itself. Thus, if $i$ is \imped in $s$, and $i$ does not change its state while the system reaches $s'$, then $s'\succeq s$ and $s'$ also does not satisfy $\mathcal{P}_{dc}$.
\end{proof}

\begin{theorem}\label{theorem:dc}
    \Cref{algorithm:dc} is a silent self-stabilizing algorithm for the dominant clique problem on $n$ nodes executing asynchronously.
\end{theorem}

\begin{proof}
    \Cref{lemma:dc} implies that if some node $i$ is \imped, then it will take an action under \Cref{algorithm:dc}. Under \Cref{algorithm:dc}, (1) if $i$ stores nodes that do not form a clique with itself, then it will first update its $i[cliq]$ variable to store only itself, i.e., $i$, and (2) if more nodes can be added to $i[cliq]$, then it will keep adding nodes to $i[cliq]$ till the nodes in $i[clique]$ form a dominant clique.

    This implies that all the nodes will eventually reach a state $s$ where they store a dominant clique that includes themselves. Therefore, \Cref{algorithm:dc} is self-stabilizing. Furthermore, in an optimal state, no node is \imped, and so no node will take any action. Thus, \Cref{algorithm:dc} is silent.
\end{proof}


To demonstrate the abstraction of the partial order induced among the local states, we define the state value of a local state and rank of a global state as follows.
$$
\begin{array}{l}
    \textsc{State-Value-DC}(i,s)=\\
        \begin{cases}
            |C|-|i[cliq]|:C \text{ is the largest superset of $i[cliq]$ that is a valid clique}
            \quad \quad \text{if $i[cliq]$ is a clique}.\\
            deg(i)+1
            \quad \quad \quad \quad \quad \quad \quad \quad \quad \quad \quad \quad \quad \quad \quad \quad \quad \quad \quad \quad \quad \quad \quad \quad \quad \quad \quad \quad  
            \text{otherwise}.
        \end{cases}
\end{array}
$$
$$\textsc{Rank-DC}(s)=\sum\limits_{i\in V(G)}\textsc{State-Value-DC}(i,s).$$

Notice that under \Cref{algorithm:dc}, in the first action of any node $i$, it is ensured that its state value is less than $deg(i)+1$. Furthermore, in all subsequent actions of $i$, its state value keeps decreasing as it keeps adding more and more nodes to $i[cliq]$ until its state value becomes 0, which is when it cannot add any more nodes. Thus, when the rank is non-zero, there is at least one \imped node in the system, and under \Cref{algorithm:dc}, the rank of the system keeps decreasing until it becomes 0, which is when no node is enabled and the system obtains a global state that is a terminal successor. Thus, $\mathcal{P}_{dc}$ is a self-stabilizing \po-inducing predicate that is utilized by \Cref{algorithm:dc}.
We provide an example to demonstrate the execution of \Cref{algorithm:dc} in Appendix \ref{appendix:dc}.

Since the nodes do not exchange any information, \Cref{algorithm:dc} does not require AMR or AA model.

\subsection{Shortest Path (\sp) Problem}\label{subsection:sp}


In this section, we present an algorithm for the shortest path problem. An algorithm for this problem \textit{cannot} be an embarrassingly parallel algorithm, unlike \Cref{algorithm:dc} for the dominant clique problem.

\begin{definition}\textbf{Shortest path.}
    In the shortest path problem, the input is a weighted arbitrary connected graph $G$ (all edge weights are positive) and a destination node $v_{des}$. Every node $i$ stores $i[d]$ (initialized with $\infty$) and $i[p]$ (initialized with $\top$). The task is to compute, $\forall i\in V(G)$, the length $i[d]$ of a shortest path from $i$ to $v_{des}$, and the parent $i[p]$ through which an entity would reach $v_{des}$ starting from $i$.
\end{definition}

In this problem, if we would have considered the local state of a node $i$ to be represented only by the variable $i[d]$ then
the local states of the nodes would form a total order. Consequently, the resultant discrete structure formed among the global states would be a $\prec$-lattice. This was shown in \cite{Garg2020}. 
On the other hand, in applications such as source routing (cf. \cite{Medhi2017})
where the source node specifies the path that should be taken, the local states form a partial order: such a system cannot be simulated within a total order.
For brevity,
we only represent the next hop, in $i[p]$.
The shortest path problem can be represented by the following predicate, where, $w(i,j)$ is the weight of edge $\{i,j\}$.
\begin{center}
    $\mathcal{P}_{sp}\equiv \forall i: (i[d]=dis(i,v_{des})=\min\{dis(j,v_{des})+w(i,j):j\in Adj_i\})\land$\quad $(i[p]=\text{arg }\min\{dis(j,v_{des})+w(i,j):j\in Adj_i\})$.
\end{center}
The local state of a node $i$ is defined by $\langle i[d], i[p]\rangle$. An \imped node $i$ in a state $s$ is a node for which its current parent is not a direct connection to the shortest path from $i$ to $v_{des}$. Formally,

\begin{center}
    $\textsc{\Imped-SP}(i)\equiv (i[d]\neq 0\land i=v_{des})\lor (\exists j\in Adj_i: i[d]>j[d]+w(i,j))$.
\end{center}

The algorithm is defined as follows. If an \imped node $i$ is $v_{des}$, then $i[d]$ is updated to 0 and $i[p]$ is updated to $v_{des}$. Otherwise, $i[p]$ is updated to the $j$ in $Adj_i$ for which $j[d]+w(i,j)$ is minimum, and $i[d]$ is updated to $j[d]+w(i,j)$.

\begin{algorithm}\label{algorithm:sp}Rules for node $i$.
\end{algorithm}
\begin{center}
    $
    \begin{array}{|l|}
        \hline 
        \textsc{\Imped-SP}(i)\longrightarrow\\
        \begin{cases}
            i[d]=0, i[p]=i & \text{if $i=v_{des}$}\\
            \langle i[d],i[p]\rangle = \langle j[d]+w(i,j), j\rangle: j = \\
            \quad \text{arg} \min\{k[d]+w(i,k):k\in Adj_i\} & \text{otherwise}
        \end{cases}~\\
        \hline 
    \end{array}
    $
\end{center}


\begin{lemma}\label{lemma:sp}
    The shortest path problem is a \po-inducing problem.
\end{lemma}


\begin{proof}
    For a node $i$, $i[d]$ contains the distance of $v_{des}$ from node $i$. A global state $s$ does not manifest all correct distances if for at least one node $i$ in $s$, (1) 
    $i$ does not store a shortest path from $i$ to $v_{des}$, or (2) the parent of $i$ is not a valid direct connection in a shortest path from $i$ to $v_{des}$.

    If a node $i$ is \imped, then there are two possibilities. First, if $i=v_{des}$ then it is not pointing to itself through $i[p]$ and $i[d]\neq 0$. Second, if $i\neq v_{des}$ then there is at least one other node $j$ such that $i[d]>j[d]+w(i,j)$ and $i[p]$ does not point to the correct node, ie., $j$.
    If $i$ is \imped in state $s$, then it violates $\mathcal{P}_{sp}$, and if 
    $i$ does not change its state while the system reaches some global state $s'$ such that $s'\succeq s$, then $i$ stays \imped in $s'$ as well, and $s'$ does not satisfy $\mathcal{P}_{sp}$.
\end{proof}

\begin{theorem}\label{theorem:sp}
    \Cref{algorithm:sp} solves the shortest path problem, on a connected positive weighted graph, on $n$ nodes executing asynchronously.
\end{theorem}


\begin{proof}
    \Cref{lemma:sp} implies that if a node $i$ is \imped, i.e., $i$ realizes that a path (through some node $j$ in $Adj_i$) has a lower cost to reach $v_{des}$ than the path that it is currently pointing to (through $i[p]$), then it will take an action and change its state under \Cref{algorithm:sp}.
    Under \Cref{algorithm:sp}, $i$ changes its local state to $\langle j[d]+w(i,j), j \rangle$. By this action, $i$ points to a node in $Adj_i$, through which $v_{des}$ can be reached via a path that is of lowest cost as computed till then.

    Since the input graph is connected, we have that system will reach a global state where all nodes find a parent through which $v_{des}$ can be reached via a shortest path, in which case, the global state will satisfy $\mathcal{P}_{sp}$.
\end{proof}

To demonstrate the abstraction of the partial order induced among the local states, we define the state value of a local state and rank of a global state as follows.
    $$
    \textsc{State-Value-SP}(i,s)=
    \begin{cases}
        \text{sum of all edge weights}+1 & \text{if $i[d]=\infty$}\\
        i[d]-dis(i,v_{des}) & \text{otherwise}
    \end{cases}
    $$
    $$\textsc{Rank-SP}(s)=\sum\limits_{i\in V(G)}\textsc{State-Value-SP}(i,s).$$

Notice that under \Cref{algorithm:sp}, in the first action of some node $i$, it is ensured that its state value is some finite value, i.e., its local state will be less than \textit{sum of all edge weights} $+1$. Furthermore, in all subsequent actions of $i$, its state value keeps decreasing as it keeps updating its local state $\langle i[d],i[p]\rangle$ until its state value becomes 0, which is when it is pointing to a parent through which it can reach $v_{des}$ via a shortest path.
Thus, when the rank is non-zero, there is at least one \imped node in the system, and under \Cref{algorithm:sp}, the rank of the system keeps decreasing until it becomes 0, which is when no node is enabled and the system obtains a global state that is a terminal successor (given that each node is initialized in the local state $\langle\infty,\top\rangle$). $P_{sp}$ is not self-stabilizing.
We provide an example to demonstrate the execution of \Cref{algorithm:sp} in Appendix \ref{appendix:sp}.

\Cref{algorithm:sp} is a distance-1 algorithm and converges even if the nodes run without synchronization in AA model. This is because 
even if a node $i$ obtains local states of another node $j$ out of order, $i$ will update its own local state only when a shorter path is found.

\subsection{Limitations of modelling problems as \po-inducing problems}\label{subsection:dag-prob-limitations}

Unlike the \po-inducing problems where the problem definition itself creates a \ldag among the global states, there are problems where the predicate $\mathcal{P}$, that defines a problem, cannot be used by the nodes to determine if they are \imped, i.e., the global states do \textit{not} form a \ldag \textit{naturally} under $\mathcal{P}$. Such problems are \textit{non-\po-inducing problems}. In such problems, there are instances in which the \imped nodes cannot be distinctly determined using $\mathcal{P}$, i.e., in those instances
$\exists s\in V(\mathcal{S}) :\lnot\mathcal{P}(s) \land (\forall i: (\exists s':\mathcal{P}(s')\land s[i]=s'[i]$)).
One such problem is maximal matching problem, defined as follows.

\begin{definition}\textbf{Maximal matching}.
    In the maximal matching problem, the input is an arbitrary graph $G$. 
    The task is to compute $i[match]$ ($\forall i\in V(G)$) such that (1) $\forall i:i[match]\neq \top\Rightarrow (i[match])[match]=i$, and (2) if $i[match]=\top$, then there must not exist a $j$ in $Adj_i$ such that $j[match]=\top$.
\end{definition}

Maximal matching (MM) is a non-\po-inducing problem.
This is because there are suboptimal states possible, in which for any given node $i$, an optimal state can be reached if $i$ does or does not change its state. Thus $i$ cannot be deemed as \imped or not \imped under the natural constraints of maximal matching.
This can be illustrated through a simple instance of a 3 nodes graph forming a simple path $\langle A,B,C\rangle$. Initially all nodes point to $\top$. Here, a maximal matching can be obtained by matching $A$ and $B$.
Thus, $C$ is not \imped. 
Another maximal matching can be obtained by matching $B$ and $C$, in which case $A$ is not \imped.
Thus the problem itself does not define which node is \imped. We call such problems non-\po-inducing problems, which we study in \Cref{section:dia}.
However,
we observe that it is possible to induce a \ldag in non-\po-inducing problems algorithmically. 

\section{Imposed \po Induction: \po-inducing \textit{Algorithms}}\label{section:dia}

In this section, we study algorithms that can be developed even for problems for which the defining predicate does not induce a \ldag among the global states.
This is because, as described in \Cref{subsection:dag-prob-limitations}, in a suboptimal global state, the problem does not specify a specific set of nodes that must change their state.

\subsection{General Properties of \po-Inducing Algorithms}\label{subsection:properties-dia}

Non-\po-inducing problems do not naturally define which node is \imped. There may be multiple optimal states. However, \imped nodes can be defined algorithmically.

\begin{definition}\label{definition:dia}\textbf{\po-inducing algorithms (PIA)}.
$A$ is a PIA for a problem $P$, and there exists a predicate $\mathcal{P}$, iff
(1) $P$ is solved iff the system reaches a state where $\mathcal{P}$ is true, and
  (2) $\mathcal{P}$ is \po-inducing with respect to $\mathcal{S}$ induced in $S$ by $A$, i.e. $\forall s\in S: \lnot\mathcal{P}(s) \Rightarrow \exists i:
     \textsc{\Imped}(i,s,\mathcal{P})$.
\end{definition}

\noindent \textbf{\textit{Remark}}: Any algorithm traversing a \ldag $\mathcal{S}$ of global states is a PIA. Thus, an algorithm that solves a \po-inducing problem, due to \po-induction, e.g. \Cref{algorithm:dc} or \Cref{algorithm:sp}, are also a PIAs.

\begin{definition}\label{definition:ss-dia}\textbf{Self-stabilizing PIA}.
    Continuing from \Cref{definition:dia}, $A$ is self-stabilizing iff in the \ldag $\mathcal{S}$ induced by $A$,
    $\forall s,s'\in S:\textsc{Terminal-Successor}(s,s')\Rightarrow\mathcal{P}(s')=true$.
\end{definition}

In the 
following, we study the maximal matching problem, a non-\po-inducing problem, and explore how a \po-inducing algorithm can be developed for such a problem. 

\subsection{Maximal Matching (MM) problem}\label{subsection:mm}

As discussed in \Cref{subsection:dag-prob-limitations}, MM is not a \po-inducing problem.
However, we develop a \po-inducing algorithm for this problem, which we elaborate in the following.
The local state of a node $i$ is defined by $\langle i[match]\rangle$.

We use the following macros. A node $i$ is \textit{wrongly matched} if $i$ is pointing to some node $j$, but $j$ is pointing to some node $k\neq i$.
A node $i$ is \textit{matchable} if $i$ is not pointing to any node, i.e. $i[match]=\top$, and there exists a node $j$ adjacent to $i$ which is also not pointing to any node.
A node $i$ is being \textit{pointed to}, or $i$ is \textit{$i$-pointed}, if $i$ is not pointing to any node,
and there exists a node $j$ adjacent to $i$ which is pointing to $i$.
A node sees that another node is being pointed, or $i$ ``sees'' \textit{else-pointed}, if some node $j$ around (in 3-hop neighbourhood of) $i$ is pointing to another node $k$ and $k$ is not pointing to anyone.
A node is \textit{unsatisfied} if it is wrongly matched or matchable. A node $i$ is \textit{\imped} if $i$ is i-pointed, or given that $i$ does not see else-pointed, $i$ is the highest ID unsatisfied node in its distance-2 neighbourhood.
\begin{center}
    \begin{tabular}{|l|}
        \hline
        $\textsc{Wrongly-Matched}(i)\equiv i[match]\neq \top\land$
        $(i[match])[match]\neq i\land (i[match])[match]\neq \top$.\\
        $\textsc{Matchable}(i)\equiv i[match]=\top\land (\exists j\in Adj_i:$
        $j[match]=\top)$.\\
        $\textsc{I-Pointed}(i)\equiv i[match]=\top\land(\exists j\in Adj_i:$ $j[match]=i)$.\\
        $\textsc{Else-Pointed}(i)\equiv \exists j\in Adj^3_i,\exists k\in Adj_j:$
        $j[match]=k\land k[match]=\top$.\\
        $\textsc{Unsatisfied}(i)\equiv\textsc{Wrongly-Matched}(i)\lor$
        $\textsc{Matchable}(i)$.\\
        $\textsc{\Imped-MM}(i) \equiv\textsc{I-Pointed}(i)\lor$
        $(\lnot \textsc{Else-Pointed}(i)\land(\textsc{Unsatisfied}(i)\land$\\
        \quad\quad\quad\quad
        $(\forall j\in Adj^2_i:i[id]>j[id]\lor\lnot\textsc{Unsatisfied}(j)))$.\\
        \hline
    \end{tabular}
\end{center}

The algorithm that we develop next is a self-stabilizing algorithm.
Thus, $i[match]$ may initially store some node that is not connected to $i$ by an edge, or some node $j$ in $Adj_i$ but $j[match]$ may not store $i$. The algorithm for an arbitrary node $i$ can be defined as follows. If $i$ is \imped and i-pointed, then $i$ starts to point to the node which is pointing at $i$. If $i$ is wrongly matched and \imped, then $i$ takes back its pointer, i.e. $i$ starts pointing to $\top$. Otherwise (if $i$ is matchable and \imped), $i$ chooses a node $j$ which is not pointing to anyone, i.e. $j[match]=\top$, and $i$ starts pointing to $j$.

\begin{algorithm}\label{algorithm:mm}
    Rules for node $i$.
\end{algorithm}
\begin{center}
    $
        \begin{array}{|l@{}|}
            \hline
            \textsc{\Imped-MM}(i)\longrightarrow\\
            \begin{cases}
                i[match]=j:j\in Adj_i: j[match]=i
                \quad \quad ~~ \text{if $\textsc{I-Pointed}(i)$}.\\
                i[match]=\top
                \quad \quad \quad \quad \quad \quad \quad \quad \quad \quad \quad \quad ~~~ \text{if $\textsc{Wrongly-Matched}(i)$}.\\
                i[match]=j:j\in Adj_i: j[match]=\top
                \quad \quad \text{otherwise}.
            \end{cases}~\\
            \hline
        \end{array}
    $
\end{center}

\begin{lemma}\label{lemma:mm}
    \Cref{algorithm:mm} induces a \ldag in the global state space under AMR model.
\end{lemma}

\begin{proof}
    Let $s$ be a suboptimal state that the input graph is in. If a node $i$ is \imped in some global state $s$, then (1) $i$ is wrongly matched, (2) $i$ is matchable, or (3) $i$ is being pointed at by another node $j$, but $i$ does not point back to $j$ or any other node. We show that if a node $i$ is \imped, and it does not change its state while the system transitions to some global state $s'$ (and so $s'\succeq s$), $s'$ will not be an optimal global state.
    
    Consider the case that some node $i$ is wrongly matched in $s$ and is \imped, and it is pointing to the same node in $s'$ as well. Then, since $i$ is \imped and $i$ is also the highest ID node that is unsatisfied, all other nodes within distance-2 of $i$ will not take an action until $i$ takes back its pointer. Thus $s'$ does not have a correct matching. 
    
    Let that in $s$, $i$ is being pointed to by some node but $i$ does not point to any node, and $i$ stays in the same state in $s'$. Let $j$ be the node that is pointing to $i$, i.e., $j[match]=i\land i[match]=\top$. Now since $\textsc{I-Pointed}(i)$ is true, so for all unsatisfied nodes within distance-2 of $i$, $\textsc{Else-Pointed}$ is true. Thus, they will not take any action until $i$ does. Also since $i[match]=\top$, $j$ will not retreat its pointer as it does not fall under the constraints of $\textsc{Wrongly-Matched}$. Thus $s'$ does not form a correct matching.
    
    Finally, in the case that $i$ is matchable and \imped in $s$, and it stays in the same state in $s'$, then it is still \imped in $s'$ as any other node in $Adj_i$ will not initiate matching with it. This is because all the unsatisfied nodes within distance-2 of $i$ have IDs less than that of $i$. Also, for the same reason, any node in $Adj^2_i$ will not take any action until $i$ updates its local state. Thus $s'$ does not manifest a maximal matching.
\end{proof}


\begin{theorem}\label{theorem:mm}
    \Cref{algorithm:mm} is a silent and self-stabilizing \po-inducing algorithm for the maximal matching problem on $n$ nodes executing asynchronously in AMR model.
\end{theorem}

\begin{proof}
    \Cref{lemma:mm} implies that if a node is \imped, then it will take an action under \Cref{algorithm:mm}.
    
    Notice that under \Cref{algorithm:mm}, by the end of the first action of any node $i$, it is ensured that it is no longer wrongly matched. After this, a node can only be either matchable or i-pointed.

    If an \imped node $i$ is matchable, then under \Cref{algorithm:mm}, it initiates a matching with a node that is not pointing to any other node. Otherwise, if a node is i-pointed, then it accepts a maching from one of its neighbouring nodes.
    
    From above, we have that a node $i$, starting from an arbitrary local state, matches with a node iff a node in $Adj_i$ becomes available to match. Therefore, the system reaches a global state where all nodes that can match with another node are matched. This shows that \Cref{algorithm:mm} is self-stabilizing. Furthermore, once an optimal state is reached, i.e., when all matchable nodes are matched, no further action is executed. This shows that \Cref{algorithm:mm} is silent.
\end{proof}

To demonstrate the abstraction of the partial order induced among the local states, we define the state value of a local state and rank of a global state as follows.

$$
\begin{array}{l}
    \textsc{State-Value-MM}(i,s)=
    \\
    \begin{cases}
        n+3 & \text{if $\textsc{Wrongly-Matched}(i)$}. \\
        |\{j\in V(G):\textsc{Wrongly-Matched}(j)\}|+3 & \text{if $\lnot\textsc{Wrongly-Matched}(i)\land\lnot\textsc{Matchable}(i)$}\\
        &\quad\quad \text{and there are some wrongly matched nodes in $G$}\\
        2 & \text{if $\textsc{Matchable}(i)\land\lnot\textsc{I-Pointed}(i)$}.\\
        1 & \text{if $\textsc{I-Pointed}$(i)}.\\
        0 & otherwise.
    \end{cases}
\end{array}
$$
$$
\begin{array}{l}
    \textsc{Rank-MM}(s)=
    \sum\limits_{i\in V(G)}\textsc{State-Value-MM}(i,s).
\end{array}
$$

Notice that under \Cref{algorithm:mm}, in the first action of any node $i$, it is ensured that its state value becomes less than $n+3$. After that a node is not wrongly matched throughout the execution of \Cref{algorithm:mm} (its state value becomes $n+2$ or less). In all subsequent actions of $i$, its state value keeps decreasing when (1) other nodes in $G$ take back their pointers if they are wrongly matched (its state value can becomes $n$ or less, and may keep decreasing), (2) node $i$ can become matchable but no one is pointing to it yet (its state value becomes 2), or (3) node $i$ initiates or accepts a matching to/from its neighbour (its state value becomes 0, or 1 and then 0). When the state value of a node becomes 0, it is properly matched, or it cannot be matched to any of its neighbours and there are no more wrongly matched nodes in $G$ who may take back their pointers. Thus, when the rank is non-zero, there is at least one \imped node in the system, and under \Cref{algorithm:mm}, the rank of the system keeps decreasing until it becomes 0, which is when no node is enabled and the system obtains a global state that is a terminal successor. Thus, a predicate that defines maximal matching is $\mathcal{P}_{mm}=\forall i:\lnot\textsc{I-pointed}(i)\land\lnot\textsc{Unsatisfied}(i)$, however the predicate that \Cref{algorithm:mm} uses to induce a \ldag among the global states and solve the problem is $\forall i\lnot \textsc{\Imped-MM}(i)$, which is a self-stabilizing predicate with respect to the \ldag induced by \Cref{algorithm:mm}.
We provide an example to demonstrate the execution of \Cref{algorithm:mm} in Appendix \ref{appendix:mm}.

\Cref{algorithm:mm} can tolerate asynchrony in AMR model, but not in AA model. This is because a node $i$ in its current state may wrongly evaluate itself to be i-pointed or unsatisfied-\imped if it gets out of order information from other nodes. This can result in $i$ to repeatedly change its state incorrectly, and consequently, the system may not reach an optimal state.

\noindent\textbf{\textit{Remark}}: In non-\po-inducing problems, e.g., \Cref{algorithm:mm}, we note that a tie-breaker is usually required to decide which nodes are \imped. Even the other algorithms in the literature have assumed similar tie-breakers (cf. \Cref{section:literature} for citations). In algorithms for \po-inducing problems, e.g., dominant clique (\Cref{subsection:dc}) and shortest path (\Cref{subsection:sp}), on the other hand, the definition of an \imped node follows the problem definition.

\section{Properties of \po-Induction}\label{section:dag-properties}

In the previous sections, we discussed example algorithms that converge even without synchronization. 
In this section, we build on that intuition and describe some characteristics of general \po-inducing systems. We study how the induction of a partial order among the local states visited by individual nodes is a necessary and sufficient condition to allow asynchrony (\Cref{subsection:dag<->asynch}). 
We also study some time-complexity properties of an algorithm that induces a \ldag among the global states (\Cref{subsection:dag-tc}).

\subsection{\po-induction to obtain asynchrony}\label{subsection:dag<->asynch}

In this subsection, we study whether \po-induction is necessary and sufficient to allow asynchronous execution. For the following, we assume that some $P$ is an arbitrary problem that requires the system to converge to a state where a predicate $\mathcal{P}$ is true. 
Let $A$ be an algorithm for $P$ that is correct under a central scheduler.

\subsubsection{Sufficiency of inducing partial order to allow asynchrony}

\begin{lemma}\label{lemma:sufficiency}
Let $\mathcal{S}$ be the transition graph that is formed under $A$, where nodes are allowed to read old values in some communication model $M$. Then,
if the local states of all individual nodes form a partial order, then $A$ guarantees convergence in asynchrony in the model $M$.
\end{lemma}

\begin{proof}
    Since the local states form a partial order, we have that under the Algorithm $A$, any node $i$ rejects each violating local state permanently. Consequently, (following \Cref{definition:dia}) a \ldag is induced among the global states.
    
    Let that some node $i$ is currently in global state $s$, but it reads stale local states of other nodes and evaluates its global state to be $s'$, then $s\succeq s'$. Since node $i$ deemed itself to be \imped in $s'$ based on old values from other nodes, its local state will inhibit the system to reach an optimal global state indefinitely, so it is also \imped in $s$. Formally, $\lnot\mathcal{P}(s')\Rightarrow \lnot\mathcal{P}(s)$ because $\textsc{\Imped}(i,s',\mathcal{P})$ and $s'[i]=s[i]$ (cf. \Cref{definition:imped-node}), thus, $i$ does not need to coordinate with other nodes to take an action and change its local state. 
\end{proof}


\subsubsection{Understanding the structure of the global state transition graph}

In \Cref{lemma:sufficiency}, we showed that \po-induction is sufficient to allow asynchrony. 
Next, we study if a partial order is guaranteed to be induced among the local states given that an algorithm is correct under asynchrony. In other words, we study if \po-induction is necessary for asynchrony.
To study that, we first examine the structure of the global state transition graph $\mathcal{S}$.

We have that the nodes run in asynchrony, so they do not necessarily have access to the fresh local states of other nodes. If, on the other hand, an apt synchronization primitive were deployed, then $\mathcal{S}$ would contain edges that would represent global state transitions that would happen only if the nodes would read fresh values from each other. However, since we are dealing with systems that must allow asynchrony, we allow all other possible transitions in $\mathcal{S}$ that may arise in the system where nodes would read old values as well. Thus, $\mathcal{S}$ does not only contain the edges that we would have assumed under a synchronization mechanism, but $\mathcal{S}$ also contains all other possible edges that may arise because of the nodes taking actions based on old values of other nodes.

\subsubsection{Necessity of inducing partial order to allow asynchrony}

Observe that in all the algorithms presented in this paper, the \textsc{State-Value} of a node $i$
provides information about how \textit{bad} the current local state of $i$ is, with respect to an optimal global state farthest from its current global state $s$.
In the following theorem, 
we show that if $\mathcal{S}$ allows asynchrony, then local states visited by all individual nodes form partial order with respect to the state values (that may utilize the local states of other nodes as well, to numerically represent how bad a local state is). 
Specifically, we provide a general algorithm to reduce from an arbitrary \dag of global states the rank of a global state, and from the rank of a global state the state values of individual nodes, ultimately showing that a partial order always exists among the local states if the global states form a \dag.


\begin{lemma}\label{lemma:necessity}
Let $\mathcal{S}$ be the transition graph that is formed under $A$, where nodes are allowed to execute asynchronously in some communication model $M$.
Then, if algorithm $A$ guarantees convergence in asynchrony under the communication model $M$, then the local states form a partial order in the model $M$. 
\end{lemma}

\begin{proof}
    Since $A$ guarantees to terminate correctly under asynchrony and the state space $S$ captures the effect of asynchrony (where a node may read old values of the variables of other nodes), there will be no cycles present among the global states in $\mathcal{S}$, i.e., $\mathcal{S}$ forms a \dag. Next, we transform $\mathcal{S}$ to a \ldag.

    Let $S_o$ be the set of optimal states.
    For each optimal global state $o\in S_o$, for each node $i$ in $o$, assign 
    the state value of $i$ to be 0. Thus, the rank of $o$ is 0 (sum of state values of all nodes).
    For every non-optimal global state $s$, we assign 
    the state value of every node $i$ to be $\top$.
    Subsequently, if all successors of state $s$ have a non-null (non-$\top$) rank (i.e., $\forall s': \langle s, s'\rangle \in E(\mathcal{S}) \Rightarrow (\forall i: \textsc{State-Value}(s'[i]) \neq \top$)) then we set the state value of each node $i$ in $s$ to be $(\max\{\textsc{Rank}(s') | \langle s, s'\rangle \in E(\mathcal{S})\})/n+1 $.
    We do this recursively until no more updates happen to any state value of any global state in $\mathcal{S}$.
    
    Since there are no cycles, this procedure will terminate in finite time. 
    Observe that using the above procedure, we obtain a valid \ldag from $\mathcal{S}$: in this \ldag, for all reachable global states $s'$ and $s''$, $s'\prec s''$ (i.e., $\langle s',s''\rangle$ is in $E(\mathcal{S})$) iff $(\forall i: s'[i]\preceq s''[i])\land(\exists i: s'[i]\prec s''[i])$, thus, following \Cref{definition:<-dag}.
    
    Therefore, we have that the local states form a partial order with respect to their assigned state values, and $A$ can be executed by all the nodes without synchronization in model $M$. 
\end{proof}

\subsubsection{Combining the necessity and sufficiency of inducing a partial order}

From \Cref{lemma:sufficiency} and \Cref{lemma:necessity}, we have

\begin{theorem}\label{theorem:dag<->partial-order}
    \noindent A concurrent algorithm $A$ guarantees convergence in asynchrony iff the local states of each node forms a partial order, i.e., iff $A$ induces a \ldag among the global states.
\end{theorem}

\Cref{theorem:dag<->partial-order} shows that induction of a partial order among local states visited by each individual node is a necessary and sufficient condition for a concurrent system to guarantee convergence in asynchrony.

\noindent\textbf{\textit{Remark}}: An algorithm $A$ guarantees convergence in some model $M$ iff the partial order is induced in $M$. E.g., \Cref{algorithm:mm} developed for maximal matching induces a \ldag in AMR model. As we discussed in \Cref{subsection:mm}, such an algorithm does not necessarily induce a \ldag in AA model, so it may not guarantee convergence in AA model.

\subsection{Time Complexity of a \texorpdfstring{\ldag}~Algorithm}\label{subsection:dag-tc}

\begin{theorem}\label{theorem:general-convergence-time}
    Given a system of $n$ processes, with the domain of (state values having) size not more than $m$ for each process, the acting algorithm will converge in $n\times (m-1)$ moves.
\end{theorem}

\begin{proof}
    Assume for contradiction that the underlying algorithm converges in $x\geq n\times (m-1)+1$ moves. This implies, by pigeonhole principle, that at least one of the nodes $i$ is revisiting their state $st$ after changing to $st'$. If $st$ to $st'$ is a step ahead transition for $i$, then $st'$ to $st$ is a step back transition for $i$ and vice versa. For a system where the global states form a \ldag, we obtain a contradiction since step-back actions are absent in such systems.
\end{proof}

\begin{corollary}\label{corollary:convergence-time-dag-multivariable}
    Let that each node $i$ stores atmost $r$ variables,  $i[var_1],...,i[var_r]$ (with domain sizes $z_1,...z_r$ respectively) contribute independently to the formation of the partial order. 
    Then an algorithm traversing the resultant \ldag will converge in $n\times ((\prod_{j=1}^r z_j)-1)$ moves.
\end{corollary}

\begin{corollary}(From \Cref{theorem:dc} and \Cref{corollary:convergence-time-dag-multivariable})
    \Cref{algorithm:dc} converges in $\sum\limits_{i\in V(G)}deg(i)=2m$ moves. In terms of rounds, it converges in $\Delta$ rounds, where $\Delta$ is the maximum degree of the input graph.
\end{corollary}

\begin{corollary}(From \Cref{theorem:sp} 
and \Cref{corollary:convergence-time-dag-multivariable})
    \Cref{algorithm:sp} converges converges in $\mathcal{D}$ rounds, where $\mathcal{D}$ is the diameter of the input graph.
\end{corollary}

\begin{corollary}\label{corollary:tc-mm}(From \Cref{theorem:mm} and \Cref{corollary:convergence-time-dag-multivariable})
    \Cref{algorithm:mm} converges in $2n$ moves.
\end{corollary}
\begin{proof}
    Any node $i$, in its first action, may take back its pointer and start pointing to $\top$. And, in its next action (the last action), $i$ will either initiate/accept a matching to/from its neighbour.
\end{proof}

\noindent\textbf{\textit{Remark}}: 
As we see in Related Work (\Cref{section:literature}), this time complexity is better than the time complexity of all the algorithms present in the literature.


\section{Related Work}\label{section:literature}

\textbf{Lattices and DAGs among global states}: In \cite{Garg2020}, the authors have described problems which possess a predicate under which the states naturally form a lattice. Definitions \ref{definition:imped-node} - \ref{definition:po-inducing-predicate} are taken from \cite{Garg2020} for total-order-inducing (TO-inducing, or lattice-linear, as called in \cite{Garg2020}) problems. However, we observe that they are equally applicable in \po-inducing systems, that is a more general framework of TO-inducing (lattice-linear) systems.
In \cite{Garg2021} and \cite{Garg2022}, the authors have studied lattice-linearity in housing market problem and some dynamic programming problems.

In \cite{Gupta2023}, the authors have studied lattice-linear problems that allow self-stabilization.
In \cite{Gupta2021}, the authors have extended the theory in \cite{Garg2020} to develop eventually lattice-linear self-stabilizing algorithms for some non-lattice-linear problems.
In \cite{Gupta2023flla}, the authors presented fully lattice-linear algorithms for several problems, where the algorithm induces single or multiple disjoint lattices among all the global states.

%

\noindent\textbf{Maximal Matching}: A distributed self-stabilizing algorithm for the maximal matching problem is presented in \cite{Hsu1992}; this algorithm was shown to converge in $O(n^3)$ moves under a central scheduler. The algorithm in \cite{Hanckowiak2001} converges in $O(\log^4n)$ moves under a synchronous scheduler. The algoritrithm for maximal matching presented in \cite{Goddard2003} converges in $n+1$ rounds under synchronous scheduler. Hedetniemi et al. (2001) \cite{Hedetniemi2001} showed that the algorithm presented in \cite{Hsu1992} converges in $2m+n$ moves.

The algorithm that we present converges in $2n$ moves and does not require synchronization. This is a significant improvement as compared to other algorithms present in the literature.


\noindent\textbf{Other Abstractions in Concurrent Computing}:
We discuss abstractions such as non-blocking (lock-free/wait-free), starvation-free, etc.
Since these models are only tangentially related to our work, we provide a comprehensive discussion on them in Appendix \ref{appendix:other-abstractions}.

\section{Conclusion}\label{section:conclusion}

We show that induction of a partial order among local states visited by all individual nodes is necessary and sufficient to allow asynchrony. Thus, we do not need to generate the entire global state transition graph and check for the absence of cycles. And so, our work reduces the computation complexity (and the complicatedness of intuition) for proving that an arbitrary algorithm is asynchrony-tolerant. In other words, our work reduces the complexity of formal proofs (and social proofs) to show asynchrony-tolerance of an algorithm.


Asynchrony-tolerant algorithms
are very useful in systems where availability and eventual convergence has greater priority than accuracy of the data, thereby allowing to avail
 the maximum possible concurrency.

We divide the problems of interest into \po-inducing and non-\po-inducing.
In \po-inducing problems, all unsatisfied nodes are enabled and can (and must) therefore evaluate their guards and take the corresponding action asynchronously. In a non-\po-inducing problem, all unsatisfied nodes are not enabled; only \imped nodes are. Thus, while studying a problem, it can be identified whether a \textit{natural} asynchrony-tolerant algorithm exists, or else a tie-breaker or another algorithmic techniques would be required to induce a \ldag among the global states.

To build the understanding of the concept, we studied some problems that fall in the classes of both \po-inducing problems and non-\po-inducing problems, and we provided example algorithms for all those problems. The algorithm that we present for maximal matching converges in significantly fewer moves as compared to other algorithms in the literature 
and can tolerate asynchrony.

We strongly insist that when a new parallel processing or distributed algorithm is designed, it should be tested if it is PO-inducing. Additionally, it will be greatly fruitful to study existing algorithms and analyze if they can tolerate asynchrony. This will help us avoid having to design new systems from scratch. Some implications of our work are discussed in Appendix \ref{appendix:implications}.

\bibliography{dag.bib}\label{bibliography}

\renewcommand{\thesubsection}{\Alph{subsection}}

\newpage 

\section*{Appendix}

\subsection{Computation Model: Execution without Synchronization}\label{appendix:computation-model}

Typically, we view the \textit{computation sequence} of an algorithm as a sequence of global states $\langle s_0, s_1, \cdots\rangle$, where $s_{\ell+1}, \ell\geq 0,$ is obtained by executing some action by one or more nodes (as decided by the scheduler) in $s_\ell$.  
For the sake of discussion, assume that only node $i$ executes in state $s_\ell$, and it has only one variable $x$. 
The computation prefix til $s_{\ell}$ is $\langle s_0, s_1, \cdots, s_\ell\rangle$. The state that the system traverses to from $s_\ell$ is $s_{\ell+1}$.
Under proper synchronization, $i$ would evaluate its guards on the \textit{current} local states of its distance-$k$ ($k\geq 1$) neighbours in $s_\ell$, resulting in the system reaching $s_{\ell+1}$.

We assume that an observer oracle is able to take a consistent snapshot of the system instantly without having to stop the executions of the nodes. However, the computing nodes may not know the fresh local states of other nodes. To understand how the execution works in asynchrony, let $s[j][x]$ be the value of a variable $x$ of some node $j$ ($j\neq i$)
in state $s$. 
If $i$ executes in asynchrony, then since the reads performed by the nodes and the movements are not coordinated with each other, the read operation performed by $i$ may return older values.
As a result, $i$ views the global state that it is in to be $s'$, 
where $s'[j][x]\in\{ s_0[j][x], s_1[j][x], \cdots, s_\ell[j][x] \}$. It means that $i$ can read older values of other nodes from arbitrarily older global states.

If all guards of $i$ evaluate to false in global state $s'$, then the system will continue to remain in state $s_\ell$, i.e., $s_{\ell+1} = s_{\ell}$.
If some guard $g$ evaluates to true in $s'$ then $i$ will execute its corresponding action $a_c$.
Here, we have the following observations:
(1) $s_{\ell+1}[i]$ is the state that $i$ obtains after executing an action in $s'$ (and not $s_\ell$), and (2) $\forall j\neq i$, $s_{\ell+1}[j] = s_\ell[j]$.

\subsection{Variations of Asynchrony}\label{appendix:asynchrony-types}

There are two orthogonal ways that can be used to achieve asynchrony in computing systems. The first approach, typically used in multiprocessor systems  (cf. \Cref{figure:2models} (a)) considers threads $p_i$ and $p_j$, that access a shared data structure. 
These threads read from the shared memory and asynchronously update it. When $p_i$ and $p_j$ execute asynchronously and they keep updating the data, $p_j$, e.g., can possibly read stale values from $p_i$. For example, $p_j$ can read from $p_i$ while $p_i$ is calculating changes to its local state and has not committed those changes yet.

This system can support primitives like CAS (compare and swap) because of the phenomenon of global availability of data through the shared memory. Additionally, in such system, if $p_i$ slows down or fails, then $p_j$ can do its work (e.g, via \textit{helping}) (cf. \cite{Adas2022}).

\begin{figure}[ht]
    \centering
    \subcaptionbox{}{
        \begin{tikzpicture}
            \node[draw=black] (data) at (0,0) {\textbf{Shared data}};
            
            \node[draw=black] (pi) at (-2.5,.5) {$p_i$};
            \node[draw=black] (pj) at (-2.5,-.5) {$p_j$};
            
            \draw[<->] (pi) -- (data.west);
            \draw[<->] (pj) -- (data.west);
            
            \node at (0,-2.5) {(CAS is allowed)};
        \end{tikzpicture}
    }
    \quad\quad 
    \subcaptionbox{}{
        \begin{tikzpicture}
            \node[draw=black] (pidata) at (0,0) {\bf $\mathbf p_i$};

            \node[draw=black, inner sep=8] (pia) at (-2.5,-1.5) {$c^i_1$};
            \node[draw=black, inner sep=8] (pib) at (-1.5,-1.5) {$c^i_2$};
            \node[draw=black, inner sep=8] (pic) at (-.5,-1.5) {$c^i_3$};
            \node (pid) at (.5,-1.5) {...};
            \node[draw=black, inner sep=8] (pie) at (1.5,-1.5) {$c^i_{n_i}$};

            \draw (pidata) -- (pia); \draw (pidata) -- (pib); \draw (pidata) -- (pic); \draw (pidata) -- (pie);

            \node at (-3.5,-1.5) {cache};

            \node[ellipse, draw, minimum width=100, minimum height=20] (copyofi) at (-2,-5) {Copy $p_i$ to $p_j$};
            
            \draw[->] (pia) -- (copyofi); \draw[->] (pib) -- (copyofi); \draw[->] (pic) -- (copyofi); \draw[->] (pie) -- (copyofi);
            
            \node[draw=black, inner sep=8] (pja) at (1,-3.5) {$c^j_1$};
            \node[draw=black, inner sep=8] (pjb) at (2,-3.5) {$c^j_1$};
            \node[draw=black, inner sep=8] (pjc) at (3,-3.5) {$c^j_1$};
            \node (pjd) at (4,-3.5) {...};
            \node[draw=black, inner sep=8] (pje) at (5,-3.5) {$c^j_{n_j}$};

            \node[draw=black] (pjdata) at (3,-5) {\bf $\mathbf p_j$};
            
            \draw (pjdata) -- (pja); \draw (pjdata) -- (pjb); \draw (pjdata) -- (pjc); \draw (pjdata) -- (pje);

            \node[ellipse, draw, minimum width=100, minimum height=20] (copyofj) at (4,0) {Copy $p_j$ to $p_i$};
            
            \draw[->] (pja) -- (copyofj); \draw[->] (pjb) -- (copyofj); \draw[->] (pjc) -- (copyofj); \draw[->] (pje) -- (copyofj);
            
            \draw[->] (copyofi) -- (pjdata);
            \draw[->] (copyofj) -- (pidata);
        \end{tikzpicture}
    }
    \caption{(a) Shared memory architecture (b) Distributed architecture.}
    \label{figure:2models}
\end{figure}

Another type of asynchronous systems consist of multiple processes, each with its own data. 
Some or all of this data is made publicly readable to other processes. This type of data sharing can be done via, e.g., a publish subscription mechanism. This data could also be put in a shared database for others to read. Process $p_j$, e.g., can maintain copies of some variables of $p_i$. 
This model is called \textit{distributed memory} model (cf. \Cref{figure:2models} (b)).

When nodes in distributed memory model execute in asynchrony, then, e.g., the copies of the variables of $p_i$ maintained by $p_j$ will be updated when $p_j$ receives new data from $p_i$ via some physical channel asynchronously. Because of this, the distributed memory model is also called \textit{message passing} model. 
This means that when $p_i$ reads its own variables, it will obtain the latest copy, however when $p_j$ reads the value of the copy of a variable held by $p_i$, it may obtain a stale copy.

This model permits each process to maintain some of its data privately. This means that only part of the data of each process is made available to others. It can support several weaker memory models (including shared memory model) such as partial store order (PSO) (cf. \cite{Xiao2022}), weak/eventual cache consistency (WCC/ECC) (cf. \cite{Afek1993}), and eventually consistent data stores (ECDS) (cf. \cite{Nguyen2019,Gonccalves2015,Attiya2017}). Such a memory model has benefits such as it support privacy and distribution of data, as well as its scalability and availability. Additionally, it allows more throughput since all the data does not have to be stored on a single machine.

We consider two variations asynchrony in shared memory and distributed memory models: 
arbitrary asynchrony (AA) and asynchrony with monotonous read (AMR). In the AA model, as described in first two paragraphs in \Cref{subsection:comcom} (as well as Appendix \ref{appendix:computation-model}), a node can read old values of other nodes arbitrarily; 
here we assume that if $p_i$ updates its variables, eventually all caches (including copies maintained at $p_j$ get updated).
Similar to AA, in AMR, the nodes execute asynchronously.
However, AMR model is arbitrary asynchrony with an additional restriction: 
a node can read the local states of other nodes only in the order of their transition.
E.g., if the state of node $j$ changes from 0 to 1 to 2 and node $i$ reads the state of $j$ to be 1 then its subsequent read will either return 1 or 2, it cannot return 0.
In the example in \Cref{figure:2models} (b), $p_j$ follows AMR model if it reads the data of $p_i$ only from a single cache, e.g., $c_1$. If it reads the data of $p_i$ from multiple caches, then since different caches may be updated at different times, it cannot be guarranteed to follow the AMR model.

Algorithms present in this paper that require AMR model are also guaranteed to converge when there is a much more relaxed requirement that given a pair of arbitrary nodes $i$ and $j$, node $i$ will eventually stop reading/receiving the states that $j$ obtained but later rejected. However, in this more relaxed model, our proofs that describe the time complexity and other characteristics of such algorithms are clearly not valid.

%
%
%

\subsection{Simple Examples}\label{appendix:colouring-and-max}

We discuss two example problems and one algorithm, each, to solve them. Our intent here is to show the difference between the working of algorithms that guarantee convergence in asynchrony, versus the algorithms that cannot make this guarantee.

First, we show how the naive algorithm for graph colouring, that we discuss in the Introduction, that does not guarantee convergence in asynchrony. 

\begin{example}
    \textbf{Colouring.}
    Consider the example algorithm for the graph colouring problem discussed in the Introduction. This algorithm, if run under a central scheduler, guarantees convergence.
    If the initial state is $\langle 1,1\rangle$, then depending on which node executes first, it is guaranteed to converge to an optimal state (see \Cref{figure:colouring-and-max} (a), solid edges going out from $\langle 1,1\rangle$).
    However, if this algorithm is run in AA or AMR model, then it does not guarantee convergence. In AA model, e.g., all possible global state transitions are possible, depending on how old the values are that the nodes are reading from each other (see \Cref{figure:colouring-and-max} (a), all dotted and solid edges).
    \qed 
\end{example}

Next, we introduce the max problem, and we discuss an algorithm for it that converges in asynchrony.

\begin{example}\label{example:max-def}
    \textbf{Max.}
    Now consider the max problem on 3 nodes, and let the initial state be $\langle 1,2,3\rangle$. Consider the following naive algorithm: if a node $i$ reads that there is another node $j$ whose value is greater than that of $i$, then $i$ changes its own value to be equal to the value of $j$.
    If this algorithm is run under a central scheduler, then convergence is guaranteed (possible transitions are the solid edges in \Cref{figure:colouring-and-max} (b)).
    Notice that if this algorithm is run in AMR or AA model, then also convergence is guaranteed.
    However, AA or AMR model result in additional transitions. Specifically, state $\langle 1,3,3\rangle$ has a predecessor $\langle 1,2,3\rangle$, so in state $\langle 1,3,3\rangle$, node 1 can read an old local state of node 2 (i.e., 2) and not the current state of node 2 (i.e., 3). Thus, another transition $\langle 1,3,3\rangle \longrightarrow \langle 2,3,3\rangle$ is possible in this case (presented as a dotted edge in \Cref{figure:colouring-and-max} (b)).
    \qed 
\end{example}

\begin{figure}[ht]
    \centering
    \subcaptionbox{}{
        \begin{tikzpicture}
            \node (a1) at (0,-3) {$\langle 1,1\rangle$};
            \node (a2) at (1.5,-1.5) {$\langle 1,2\rangle$};
            \node (a3) at (-1.5,-1.5) {$\langle 2,1\rangle$};
            \node (a4) at (0,0) {$\langle 2,2\rangle$};
            
            \draw[->] (a1) -- (a2); \draw[->] (a1) -- (a3);
            \draw[dashed, ->] (a2.south) -- (a1.east); \draw[dashed, ->] (a3.south) -- (a1.west);
            \draw[->] (a4) -- (a2); \draw[->] (a4) -- (a3);
            \draw[dashed, ->] (a2.north) -- (a4.east); \draw[dashed, ->] (a3.north) -- (a4.west);
            \draw[dashed, ->] (.2,-.25) -- (.2,-2.75); \draw[dashed, ->] (-.2,-2.75) -- (-.2,-.25);
        \end{tikzpicture}
    }
    \subcaptionbox{}{
        \begin{tikzpicture}
            \node (a123) at (0,0) {$\langle 1,2,3\rangle$};
            \node (a133) at (-2,1) {$\langle 1,3,3\rangle$};
            \node (a223) at (0,1) {$\langle 2,2,3\rangle$};
            \node (a323) at (2,2) {$\langle 3,2,3\rangle$};
            \node (a233) at (0,2) {$\langle 2,3,3\rangle$};
            \node (a333) at (0,3) {$\langle 3,3,3\rangle$};
            
            \draw[->] (a123) -- (a133); \draw[->] (a123) -- (a223); \draw[->] (a123) -- (a323);
            \draw[->] (a223) -- (a233);
            \draw[->] (a233) -- (a333);
            \draw[->] (a223) -- (a323);
            \draw[->] (a323) -- (a333);
            \draw[->] (a133) -- (a333);
            
            \draw[dashed, ->] (a133) -- (a233);
        \end{tikzpicture}
    }
    \caption{See Appendix \ref{appendix:colouring-and-max} for explanation. (a) State transitions of a naive algorithm for graph colouring. (b) State transitions under a naive algorithm for the max problem.}
    \label{figure:colouring-and-max}
\end{figure}

\subsection{Examples following definitions}\label{appendix:examples-following-definitions}

\noindent\textbf{Examples following \Cref{definition:imped-node}}

\begin{examplemaxcont}
    Under the algorithm for the max problem presented in \Cref{example:max-def} (Appendix \ref{appendix:colouring-and-max}), we notice that in any suboptimal global state $s$, an impedensable node $i$ is a node that does not have the highest value stored in it. For example in state $\langle 2,2,3\rangle$ node 1 and node 2 are \imped. If any of these nodes retains its state, then it will prevent the system from convergence.
    \qed 
\end{examplemaxcont}

\begin{examplemaxcont}
    Consider the execution of the nodes in global state $\langle 2,2,3\rangle$ under the algorithm for the max problem presented in Appendix \ref{appendix:colouring-and-max}. $\textsc{Equicorrupt}(2,1)$ is the set of local states of node 1 that are equally corrupt as local state 2. There is only one local state in this set, i.e., $\textsc{Equicorrupt}(2,1)=\{2\}$.
    \qed 
\end{examplemaxcont}

\begin{examplegccont}
    Notice that in the case of graph colouring, an \imped node cannot be determined. This is because, in any given suboptimal global state, for any chosen node, convergence can be achieved without changing the local state of that node.
\end{examplegccont}

\noindent\textbf{Examples following \Cref{definition:imped-state}}

\begin{examplemaxcont}
    Under the algorithm for the max problem presented in Appendix \ref{appendix:colouring-and-max}, we notice that in any suboptimal global state, there is at least one \imped node. All suboptimal global states are \imped global states.
    \qed 
\end{examplemaxcont}

\begin{examplemaxcont}
    Under the algorithm for the max problem presented in Appendix \ref{appendix:colouring-and-max}, the local states form a partial order, present in \Cref{figure:max-po}. However, e.g., since node 2 is initialized in local state $2$, the order induced among its local states is $2 \rightarrow 3$.
    
    Since all nodes follow the same algorithm, the partial order formed among the local states is essentially the same; its starting point only depends on the local state of initialization for each individual node.
    
    Due to the partial order formed among the local states of each individual node, the global states form a \ldag, which is shown in \Cref{figure:colouring-and-max} (b).
    \qed 
\end{examplemaxcont}

\begin{figure}[ht]
    \centering
    \begin{tikzpicture}
        \node (a1) at (0,0) {$1$};
        \node (a2) at (0,1) {$2$};
        \node (a3) at (1.5,2) {$3$};
        
        \draw[->] (a1) -- (a2); \draw[->] (a2) -- (a3); \draw[->] (a1) -- (a3);
    \end{tikzpicture}
    \caption{Partial order formed among the local states in the algorithm for max problem (\Cref{example:max-def}).}
    \label{figure:max-po}
\end{figure}

\noindent\textbf{Examples following \Cref{definition:po-inducing-predicate}}

\begin{examplemaxcont}
    \label{example:max-predicate}
    The predicate, governing the algorithm for the max problem (cf. Appendix \ref{appendix:colouring-and-max}) can be noted as follows.
    $$\mathcal{P}_{max}\equiv \forall i,j: j[val]\not>i[val].$$
    From $\mathcal{P}_{max}$, we can derive the definition of an \imped node. Specifically, the predicate defining an \imped node is noted as follows.
    $$\textsc{\Imped-Max}(i)\equiv \exists j: j[val]>i[val].$$
    When this predicate is combined with the actions imposed by the algorithm, the global state transition graph, as presented in \Cref{figure:colouring-and-max} (b), is formed as a \ldag.
    \qed 
\end{examplemaxcont}

\noindent\textbf{Examples following \Cref{definition:terminal-successors}}

\begin{examplemaxcont}
    The only terminal successor in the state transition graph shown in \Cref{figure:colouring-and-max} (b) is $\langle 3,3,3\rangle$.
    \qed 
\end{examplemaxcont}

\noindent\textbf{Examples following \Cref{definition:ss-dip}}

\begin{examplemaxcont}
    It can be 
    noticed that the predicate for the max problem, as noted in \Cref{example:max-predicate} is a self-stabilizing \po-inducing predicate.
    \qed 
\end{examplemaxcont}

\noindent\textbf{Examples following \Cref{definition:dip}}

\begin{examplemaxcont}
    The max problem, as discussed in Appendix \ref{appendix:colouring-and-max} is a \po-inducing problem. This is because, in any suboptimal state, we can determine all the nodes that are \imped.
    \qed 
\end{examplemaxcont}

\begin{examplegccont}
    Since an \imped node cannot be determined in the graph colouring problem, we put the graph colouring problem in the class of non-\po-inducing problems. We discuss more on such problems in \Cref{subsection:dag-prob-limitations} and \Cref{section:dia}.
    \qed 
\end{examplegccont}

\subsection{An Example for \Cref{algorithm:dc}}\label{appendix:dc}


We provide a figure to show how nodes take action under \Cref{algorithm:dc}: the partial order induced among the local states of node $v_1$ (of the graph in \Cref{figure:dc} (a)) is shown in \Cref{figure:dc} (b).

\begin{figure}[ht]
    \centering
    \subcaptionbox{}{
        \begin{tikzpicture}
            \node [circle, draw=black,fill=black,inner sep=2pt,label=above:$v_1$] (a) at (0,0) {};
            \node [circle, draw=black,fill=black,inner sep=2pt,label=left:$v_3$] (b) at (-.5,-1) {};
            \node [circle, draw=black,fill=black,inner sep=2pt,label=below:$v_2$] (c) at (.5,-1) {};
            
            \draw (a) -- (b); \draw (a) -- (c);
        \end{tikzpicture}
    }
    \subcaptionbox{}{
        \begin{tikzpicture}
            \node (a) at (0,0) {\begin{tabular}{c}\{$v_1$\}\end{tabular}};
            \node (b) at (-1,1) {\begin{tabular}{c}\{$v_1$,$v_2$\}\end{tabular}};
            \node (c) at (1,1) {\begin{tabular}{c}\{$v_1$,$v_3$\}\end{tabular}};
            
            \node (d) at (-3,-1) {\begin{tabular}{c}\{$v_1$,$v_2$,$v_3$\}\end{tabular}};
            \node (e) at (-1.25,-1) {\begin{tabular}{c}\{$v_2$,$v_3$\}\end{tabular}};
            \node (f) at (0,-1) {\begin{tabular}{c}\{\}\end{tabular}};
            \node (g) at (1,-1) {\begin{tabular}{c}\{$v_2$\}\end{tabular}};
            \node (h) at (2,-1) {\begin{tabular}{c}\{$v_3$\}\end{tabular}};
            \draw (a) -- (b); \draw (a) -- (c);\draw (a) -- (d); \draw (a) -- (e); \draw (a) -- (f); \draw (a) -- (g); \draw (a) -- (h);
            
            \node at (3.25,1) {\begin{tabular}{c}s.v.=0\end{tabular}};
            \node at (3.25,0) {\begin{tabular}{c}s.v.=1\end{tabular}};
            \node at (3.25,-1) {\begin{tabular}{c}s.v.=3\end{tabular}};
        \end{tikzpicture}
    }
    \subcaptionbox{}{
        \begin{tikzpicture}[scale=.78,every node/.style={scale=.78},y=1.4cm]
            \node (a) at (0,-1.5) {\begin{tabular}{c}$\langle$\{1\},\{2\},\{3\}$\rangle$\\ $\langle 1,1,1\rangle=3$\end{tabular}};
            
            \node (b1) at (-4.5,0) {\begin{tabular}{c}$\langle$\{1,2\},\{2\},\{3\}$\rangle$\\ $\langle 0,1,1\rangle=2$\end{tabular}};
            \node (b2) at (-1.5,0) {\begin{tabular}{c}$\langle$\{1\},\{1,2\},\{3\}$\rangle$\\ $\langle 1,0,1\rangle=2$\end{tabular}};
            \node (b3) at (1.5,0) {\begin{tabular}{c}$\langle$\{1\},\{2\},\{1,3\}$\rangle$\\ $\langle 1,1,0\rangle=2$\end{tabular}};
            \node (b4) at (4.5,0) {\begin{tabular}{c}$\langle$\{1,3\},\{2\},\{3\}$\rangle$\\ $\langle 0,1,1\rangle=2$\end{tabular}};

            \node (c1) at (-7.5,2) {\begin{tabular}{c}$\langle$\{1,2\},\{1,2\},\{3\}$\rangle$\\ $\langle 0,0,1\rangle=1$\end{tabular}};
            \node (c2) at (-4.5,2) {\begin{tabular}{c}$\langle\{1,2\},\{2\},\{1,3\}\rangle$\\ $\langle 0,1,0\rangle=1$\end{tabular}};
            \node (c3) at (-1.5,2) {\begin{tabular}{c}$\langle\{1\},\{1,2\},\{1,3\}\rangle$\\ $\langle 0,0,1\rangle=1$\end{tabular}};

            \node (c6) at (1.5,2) {\begin{tabular}{c}$\langle$\{1\},\{1,2\},\{1,3\}$\rangle$\\ $\langle 1,0,0\rangle=1$\end{tabular}};
            \node (c4) at (4.5,2) {\begin{tabular}{c}$\langle$\{1,3\},\{1,2\},\{3\}$\rangle$\\ $\langle 0,0,1\rangle=1$\end{tabular}};
            \node (c5) at (7.5,2) {\begin{tabular}{c}$\langle\{1,3\},\{2\},\{1,3\}\rangle$\\ $\langle 0,1,0\rangle=1$\end{tabular}};
            
            \node (d1) at (-2,3.5) {\begin{tabular}{c}$\langle$\{1,2\},\{1,2\},\{1,3\}$\rangle$\\ $\langle 0,0,0\rangle=0$\end{tabular}};
            \node (d2) at (2,3.5) {\begin{tabular}{c}$\langle$\{1,3\},\{1,2\},\{1,3\}$\rangle$\\ $\langle 0,0,0\rangle=0$\end{tabular}};
            
            \draw (a) -- (b1); \draw (a) -- (b2);\draw (a) -- (b3);\draw (a) -- (b4);
            
            \draw (b1) -- (c1);\draw (b2) -- (c1);
            \draw (b1) -- (c2);\draw (b3) -- (c2);
            \draw (b2) -- (c3);\draw (b3) -- (c3);

            \draw (b2) -- (c4);\draw (b4) -- (c4);
            \draw (b3) -- (c5);\draw (b4) -- (c5);
            \draw (b2) -- (c6);\draw (b3) -- (c6);
            
            \draw (c1) -- (d1);\draw (c2) -- (d1);\draw (c3) -- (d1);
            \draw (c4) -- (d2);\draw (c5) -- (d2);\draw (c6) -- (d2);

        \end{tikzpicture}
    }
    \caption{(a) Input graph. (b) Partial order among local states of node 1 (The acronym `s.v.' stands for state value; all edges are directed upwards). (c) \ldag, assuming that initial state is $\langle\{1\},\{2\},\{3\}\rangle$; we replaced writing $v_i$ by $i$ for brevity. In all these states, the second guard of \Cref{algorithm:dc} is false. In every data point, the first row shows the global state, the second row shows the respective state values of nodes and the rank of the global state. 
    (Transitive edges are not shown; all edges are directed upwards.)}
    \label{figure:dc}
\end{figure}


Under \Cref{algorithm:dc}, the global states of the graph in \Cref{figure:dc} (a) form a \ldag that we show in \Cref{figure:dc} (c). For a pair of global states $s$ and $s'$, if $s\prec s'$ then $\textsc{Rank-DC}(s')<\textsc{Rank-DC}(s)$.
In \Cref{figure:dc} (c), a global state is represented as $\langle\langle v_1[cliq]\rangle,$ $\langle v_2[cliq]\rangle,$ $\langle v_3[cliq]\rangle\rangle$.
The state space for this example graph has a total of 512 global states. In the figure, we only show the states where the second guard is false in all the nodes. All the global states where the second guard is true in some nodes will converge to one of the global states present in this figure, and then it will converge to one of the terminal successors.

\subsection{
An example for \Cref{algorithm:sp}}\label{appendix:sp}

Under \Cref{algorithm:sp}, the global states form a \ldag. We show an example in \Cref{figure:sp-global-states}.
\Cref{figure:sp-global-states} (a) is the input graph and \Cref{figure:sp-global-states} (b) is the \ldag induced among the global states. For a pair of global states $s$ and $s'$, if $s\prec s'$ then $\textsc{Rank-SP}(s')<\textsc{Rank-SP}(s)$. In \Cref{figure:sp-global-states}, a global state is represented as $\langle\langle v_1[p]$, $v_1[d]\rangle$, $...$, $\langle v_4[p]$, $v_4[d]\rangle\rangle$.

\begin{figure}[ht]
    \centering
    \subcaptionbox{}{
        \begin{tikzpicture}
            \node [circle, inner sep=2pt, fill=black, draw=black, label=above:$v_1$] (v1) at (0,0) {};
            \node [circle, inner sep=2pt, fill=black, draw=black, label=left:$v_2$] (v2) at (-1,-1) {};
            \node [circle, inner sep=2pt, fill=black, draw=black, label=right:$v_3$] (v3) at (1,-1) {};
            \node [circle, inner sep=2pt, fill=black, draw=black, label=below:{$v_4=v_{des}$}] (v4) at (0,-2) {};
            
            \draw (v4) -- node[below] {2} (v2); \draw (v4) -- node[below] {1} (v3);
            \draw (v1) -- node[above] {2} (v2); \draw (v1) -- node[above] {3} (v3);
        \end{tikzpicture}
    }
    \subcaptionbox{}{
        \begin{tikzpicture}
            \node (a1) at (0,0) {\begin{tabular}{c}$\langle\langle\top,\infty\rangle$,$\langle\top,\infty\rangle$,\\$\langle\top,\infty\rangle$,$\langle\top,\infty\rangle\rangle$\\$\langle10,10,10,10\rangle=40$\end{tabular}};
            \node (a2) at (0,2) {\begin{tabular}{c}$\langle\langle\top,\infty\rangle$,$\langle\top,\infty\rangle$,\\$\langle\top,\infty\rangle$,$\langle v_4,0\rangle\rangle$\\$\langle10,10,10,0\rangle=30$\end{tabular}};
            \node (a3) at (-2,4) {\begin{tabular}{c}$\langle\langle\top,\infty\rangle$,$\langle v_4,2\rangle$,\\$\langle\top,\infty\rangle$,$\langle v_4,0\rangle\rangle$\\$\langle10,0,10,0\rangle=20$\end{tabular}};
            \node (a4) at (2,4) {\begin{tabular}{c}$\langle\langle\top,\infty\rangle$,$\langle\top,\infty\rangle$,\\$\langle v_4,1\rangle$,$\langle v_4,0\rangle\rangle$\\$\langle10,10,0,0\rangle=20$\end{tabular}};
            \node (a5) at (0,6) {\begin{tabular}{c}$\langle\langle\top,\infty\rangle$,$\langle v_4,2\rangle$,\\$\langle v_4,1\rangle$,$\langle v_4,0\rangle\rangle$\\$\langle10,0,0,0\rangle=10$\end{tabular}};
            \node (a6) at (-2,8) {\begin{tabular}{c}$\langle\langle v_2,4\rangle$,$\langle v_4,2\rangle$,\\$\langle v_4,1\rangle$,$\langle v_4,0\rangle\rangle$\\$\langle 0,0,0,0\rangle=0$\end{tabular}};
            \node (a7) at (2,8) {\begin{tabular}{c}$\langle\langle v_3,4\rangle$,$\langle v_4,2\rangle$,\\$\langle v_4,1\rangle$,$\langle v_4,0\rangle\rangle$\\$\langle 0,0,0,0\rangle=0$\end{tabular}};
            
            \draw (a1) -- (a2);
            \draw (a2) -- (a3); \draw (a2) -- (a4);
            \draw (a3) -- (a5); \draw (a3) -- (a6);
            \draw (a4) -- (a5); \draw (a4) -- (a7);
            \draw (a5) -- (a6); \draw (a5) -- (a7);
        \end{tikzpicture}
    }
    \caption{(a) Input graph. (b) \ldag induced among the global states in evaluating for the shortest path parblem in the graph shown in (a); a global state is represented as $\langle\langle v_1[p],v_1[d]\rangle,...,\langle v_4[p],v_4[d]\rangle\rangle$. In every data point, the first two rows show the global state, the third row shows the respective state values of nodes and the rank of the global state. (Transitive edges are not shown, all edges are directed upwards).
    }
    \label{figure:sp-global-states}
\end{figure}

\subsection{An example for \Cref{algorithm:mm}}\label{appendix:mm}

Under \Cref{algorithm:mm}, the global states form a \ldag. We show an example in \Cref{figure:mm-global-states}: \Cref{figure:mm-global-states} (a) is the input graph and \Cref{figure:mm-global-states} (b) is the induced \ldag. For a pair of global states $s$ and $s'$, if $s\prec s'$ then $\textsc{Rank-MM}(s')<\textsc{Rank-MM}(s)$. In \Cref{figure:mm-global-states} (b), a global state is represented as $\langle\langle v_1[match]\rangle$, $...$, $\langle v_4[match]\rangle\rangle$.
Also, the global state transition diagram present in \Cref{figure:mm-global-states} (b) is drawn assuming that the initial global state is $\langle \top,\top,\top,\top\rangle$. Observe that any other state will converge to one of the states drawn in this \ldag and then converge to one of the optimal states.

\begin{figure}[ht]
    \centering
    \subcaptionbox{}{
        \begin{tikzpicture}
            \node [circle, draw=black,fill=black,inner sep=2pt,label=above:$v_4$] (a) at (0,0) {};
            \node [circle, draw=black,fill=black,inner sep=2pt,label=left:$v_2$] (b) at (-.5,-1) {};
            \node [circle, draw=black,fill=black,inner sep=2pt,label=below:$v_3$] (c) at (.5,-1) {};
            \node [circle, draw=black,fill=black,inner sep=2pt,label=above:$v_1$] (d) at (1,0) {};
            
            \draw (a) -- (b); \draw (a) -- (c);\draw (a) -- (d);
        \end{tikzpicture}
    }
    \subcaptionbox{}{
        \begin{tikzpicture}
            \node (a) at (0,0) {\begin{tabular}{c}$\langle \top,\top,\top,\top\rangle$\\$\langle 2,2,2,2\rangle=8$\end{tabular}};
            \node (b) at (-3,1.5) {\begin{tabular}{c}$\langle \top,\top,\top,v_1\rangle$\\$\langle 1,0,0,0\rangle=1$\end{tabular}};
            \node (c) at (0,1.5) {\begin{tabular}{c}$\langle \top,\top,\top,v_2\rangle$\\$\langle 0,1,0,0\rangle=1$\end{tabular}};
            \node (d) at (3,1.5) {\begin{tabular}{c}$\langle \top,\top,\top,v_3\rangle$\\$\langle 0,0,1,0\rangle=1$\end{tabular}};
            
            \node (b2) at (-3,3) {\begin{tabular}{c}$\langle v_4,\top,\top,v_1\rangle$\\$\langle 0,0,0,0\rangle=0$\end{tabular}};
            \node (c2) at (0,3) {\begin{tabular}{c}$\langle \top,v_4,\top,v_2\rangle$\\$\langle 0,0,0,0\rangle=0$\end{tabular}};
            \node (d2) at (3,3) {\begin{tabular}{c}$\langle \top,\top,v_4,v_3\rangle$\\$\langle 0,0,0,0\rangle=0$\end{tabular}};
            \draw (a) -- (b); \draw (a) -- (c);\draw (a) -- (d);
            \draw (b2) -- (b); \draw (c2) -- (c);\draw (d2) -- (d);
        \end{tikzpicture}
    }
    \caption{(a) Input graph. (b) The state transition diagram, a \ldag, assuming that the initial global state is $\langle \top,\top,\top,\top\rangle$. In every data point, the first row shows the global state, the second row shows the respective state values of nodes and the rank of the global state.
    (All edges are directed upwards; transitive edges are not shown for brevity).}
    \label{figure:mm-global-states}
\end{figure}


\subsection{Other Related Abstractions in Concurrent Computing}\label{appendix:other-abstractions}

Below, we discuss some models related to synchronization that are tangentially related to the paper. 

An algorithm is \textit{non-blocking} if in a system running such algorithm, if a node fails or is suspended, then it does not result in failure or suspension of another node. 
A non-blocking algorithm is \textit{lock-free} if system-wide progress can be guaranteed. For example, algorithms for implementing lock-free singly-linked lists and binary search tree are, respectively, presented in \cite{Valois1995} and \cite{Natarajan2014}. In such systems, if a read/write request is blocked then other nodes continue their actions normally.

A non-blocking algorithm is \textit{wait-free} if progress can be guaranteed per node. A wait-free sorting algorithm is studied in \cite{Shavit1997}, which sorts an array of size $N$ using $n\leq N$ computing nodes, and an $O(n)$ time wait-free approximate agreement algorithm is presented in \cite{Attiya1994}. In such systems, in contrast to lock-free systems, it must be guaranteed that all nodes make progress individually.


As it can be seen from the algorithms presented in this paper, if a node crashes, then the system may not reach an optimal state. This is true for lattice-linear algorithms present in the literature as well (e.g., algorithms in \cite{Garg2020,Gupta2023flla} and others). This is because in these systems, each node simulates a separate entity (e.g., a distinct vertex of the input graph $G$), and a node works representing only that entity. These systems are specifically designed to require that all nodes reach a local state that satisfy certain constraints. And, if a node crashes, e.g., then the system may not reach an optimal state. Thus lock-freedom and wait-freedom cannot be expected from these systems because of the way these systems are designed.

That said, the goal of asynchronous algorithms is not the progress / blocking of individual nodes 
(e.g., success of insert request in a linked list and a binary search tree, respectively, in \cite{Valois1995} and \cite{Natarajan2014}).
Rather it focuses on the progress from the perspective of the system, i.e., the goal is not about the progress of an action by a node but rather that of the entire system.  
For example, in the algorithm for maximal matching present in this paper (\Cref{algorithm:mm} in \Cref{subsection:mm}), if an \imped node is slow or does not move, the system will not converge until it takes an action. 
However, the nodes can run without any coordination and they can execute based on old values, instead of requiring a synchronization primitive to ensure convergence.
In fact, the notion of \imped (recall that in the algorithms that we present in this paper, in any global state, all enabled nodes are \imped and vice versa) captures this. An \imped node has to make progress in order for the system to make progress.

\textit{Starvation} happens when requests of a higher priority prevent a request of lower priority from entering the critical section indefinitely. To prevent starvation, algorithms are designed such that the priority of pending requests are increased dynamically. Consequently, a low-priority request eventually obtains the highest priority. Such algorithms are called \textit{starvation-free} algorithms. For example, authors of \cite{Kim2005} and \cite{Attiya2010}, respectively, present a starvation-free algorithm to schedule queued traffic in a network switch and a starvation-free distributed directory algorithm for shared objects.
Asynchronous algorithms are starvation-free, as long as all enabled nodes do not die and can execute. If all enabled processes can execute, convergence is guaranteed.

\textit{Serializability} allows only those executions to be executed concurrently which can be modelled as some permutation of a sequence of those executions. In other words, serializability does not allow nodes to read and execute on old information of each other: only those executions are allowed in concurrency such that reading fresh information, as if the nodes were executing under a central scheduler, would give the same result. Serializability is heavily utilized in database systems, and thus, the executions performed in such systems are called \textit{transactions}. Authors of \cite{Papadimitriou1979} show that corresponding to several transactions, determining whether a sequence of read and write operations is serializable is an NP-Complete problem. They also present some polynomial time algorithms that approximate such serializability. Authors of \cite{Fle1982} consider the problem in which the sequence of operations performed by a transaction may be repeated infinitely often. They describe a synchronization algorithm allowing only those schedules that are serializable in the order of commitment.

The asynchronous execution considered in this paper is not \textit{serializable}, especially, since the reads can be from an old global state. Even so, the algorithm converges, and does not suffer from the overhead of synchronization required for serializability.

In \textit{redblue} systems (e.g., \cite{Li2012}), the rules can be divided into two non-empty sets: red rules, which must be synchronized, and blue rules, which can run in a lazy manner and do not have to be synchronized. PO-inducing systems in general are the systems in which red rules are absent as an enabled node can execute independently without any synchronization.

In \textit{local mutual exclusion}, at a given time, some nodes block other nodes while entering to critical section. This can be done, e.g., by deploying semaphores.
Authors of \cite{Raymond1989} presented an algorithm for distributed mutual exclusion in computer networks, that uses a spanning tree of the subject network. In this algorithm, the number of messages exchanged per critical section depends on the topology of this tree, typically this value is $O(n)$.
Authors of \cite{Yang1995} an algorithm with $O(\lg n)$ time complexity for mutual exclusion among $n$ nodes. Specifically, this algorithm requires atomic reads and writes and in which all spins are local (here a spin means a busy wait in which a node, in this case, waits on locally accessible shared variables).

We see, in algorithms based on local mutual exclusion, that they require additional data structures/variables to ensure that access is provided to (and blocking is deployed on) a certain set of processes. In asynchronous algorithms, nodes do not block each other. In non-PO-inducing problems, we see that usually a tie-breaker is required to ensure the correctness of the executions, however, if a problem is naturally PO-inducing, then it is not required. This is because in the case of non-PO-inducing problems, it may be desired that all unsatisfied nodes do not become enabled, however, in the case of PO-inducing problems, as we see in \cite{Garg2020}, and in \Cref{section:dip} in this paper, all unsatisfied nodes are be enabled. And, all enabled nodes can read values and perform executions asynchronously, where they are allowed to read old values, which is not allowed in algorithms that deploy mutual exclusion.

\subsection{Implications of our theory}\label{appendix:implications}

In this paper, we showed that 
local state transitions being abstracted as a partial order
is both necessary and sufficient for an algorithm to allow asynchrony. Focusing on the \textit{necessary} condition, we find that if there is any existing technique that permits the program to be executed under asynchrony, then the local states of each individual node form a partial order. 

A large number of problems can be solved by algorithms that may converge without requiring any synchronization. A high fraction of such algorithms may already be published, but it has not been proven so, meaning that the authors have made stricter assumptions while analyzing those algorithms. Recently, there has been work that shows for some existing algorithms that they do not require synchronization. (We discuss this later in this section.)

Any algorithm that can converge in asynchrony monotonously reduces its rank with respect to the farthest optimal state. However, in many such algorithms, it is not directly apparent that they have such a property. There can be several issues that make it difficult to show that (and if) an algorithm is tolerant to asynchrony. We have covered some of those issues through the problems that we studied in this paper. For example:
\begin{enumerate}
    \item A system may store a local state that is not numerical, e.g., the \Cref{algorithm:dc} for the \dc problem.
    \item A system may store the local states that is numerical and may monotonously decrease, but not reach zero, e.g., the \sp problem (cf. \Cref{algorithm:sp}, \Cref{subsection:sp}).
    \item A system may have a property in which the nodes have the potential to revisit their local state, but
    it still guarantees convergence even under asynchrony, e.g., \Cref{algorithm:mm} for the maximal matching problem.
\end{enumerate}
Despite all such differences and behaviours that are difficult to analyze, these are some examples of algorithms that guarantee convergence without any synchronization.
For all these algorithms, we describe an abstraction, by defining state-values and ranks, under which they form a \ldag.

There are other implications of the theory that we presented in this paper. 
One immediate implication of our theory is in its application in writing proofs. Developing asynchronous systems has been difficult, and one of the reasons that adds to the intricacy is writing proofs of correctness of such systems. Our theory not only insists on the possibility of simplifying such proofs, but also provides an upper bound to the time complexity of the runtime of an arbitrary asynchronous algorithm. \Cref{theorem:dag<->partial-order} implies that to show that an algorithm is tolerant to asynchrony, we only need to show that the local states visited by individual nodes can be abstracted as a partial order, rather than to generate the entire state space and checking for the existence of a cycle. Thus, our theory has the potential to simplify the proofs that show tolerance to asynchrony. In addition, \Cref{corollary:convergence-time-dag-multivariable} provides the upper bound of the time complexity of a \po-inducing algorithm.

In recent literature, existing algorithms have been shown to be tolerant to asynchrony by only showing that a partial (or a total) order exists among the local states visited by individual nodes. For example, Johnson's algorithm for computing shortest paths \cite{Johnson1977}, Gale-Shapey algorithm for stable marriage \cite{Gale1962}, Cesari-Maeder parallelization \cite{Cesari1996} of Karatsuba's multiplication, GSGS algorithm for converging myopic robots on an infinite triangular grid \cite{Goswami2022}, have been shown to be tolerant to asynchrony, by showing that the local state transitions can be abstracted as a total order, in respectively, \cite{Garg2020}, \cite{Garg2020}, \cite{Gupta2023} and \cite{Gupta2024}. This is a consequence of the observation that these algorithms stipulate that all \imped nodes must update their local states.
Our paper studies a general case of the conditions of lattice-linearity studied in these papers, where we study algorithms in which the local states of nodes form a partial order. Interestingly, we find that any algorithm that is tolerant to asynchrony satisfies this condition.

Our theory also implies that some classes of algorithms are tolerant to asynchrony.
One such instance of asynchronous execution encapsulates a special case of the read-write model \cite{Dolev1993}, where reads and writes do not depend on a scheduler or any other similar global condition.
In the \textit{read-write model}, an action is either a read action (read a variable from one other process (and save its local copy)) or a write action (write your own variable (based on those copies)). Consider the situation where the read action in such a program can be executed at any time (i.e., without depending on a scheduler or any other condition). Under such a scenario, let us focus on the execution of the write action. In the write action of some node $i$, it is relying on value of variables of some other nodes, say $j$ and $k$, that it read previously. Here, the state of $j$ and $k$ may have already changed. Furthermore the values read for $j$ and $k$ may not be consistent with each other. In other words, such a program satisfies its specification then it is doing so while permitting asynchronous execution considered in this paper. The necessity result in this paper implies that a \ldag is induced in the state transition graph. 

A message-passing program also inherently appears to permit asynchronous execution, in the cases where the executions are independent of a scheduler or another condition. In the message-passing model, the nodes are placed remotely and so, nodes rely on messages received in the past. Consider a message passing program where whenever a node updates its state, it forwards its new state to all its neighbors. 
Additionally, assume that the message channels are FIFO. In such a program, when node $i$ changes its state based on the messages received from others, say $j$ and $k$, it is doing so with old information about $j$ and $k$, as this information may have changed in the interim. In other words, such an implementation satisfies its specification then it is doing so while permitting asynchronous execution considered in this paper. The necessity result in this paper implies that a \ldag is always being induced in the global state transition graph with the corresponding \imped condition.

\end{document}